\documentclass[12pt]{article}

\usepackage{amsmath, amsfonts, amssymb, amsthm, mathtools}

\usepackage[margin=1in]{geometry}
\usepackage{setspace}
\usepackage{eurosym}
\usepackage{upgreek}
\usepackage{tipa}
\usepackage{xcolor}

\usepackage[
  colorlinks=true,
  linkcolor=red,
  citecolor=blue,
  urlcolor=blue
]{hyperref}

\usepackage{graphicx}
\usepackage{float}
\usepackage{array}
\usepackage{multirow}
\usepackage{pdflscape}
\usepackage{subfigure}
\usepackage{caption}

\usepackage{sectsty}
\usepackage{comment}
\usepackage{footmisc}

\usepackage{enumitem}

\usepackage{tikz}
\usetikzlibrary{decorations.pathreplacing}

\usepackage[normalem]{ulem}

\usepackage{natbib}  

\usepackage[colorlinks=true, urlcolor=blue]{hyperref}

\subsubsectionfont{\normalfont\itshape}

\newcommand{\ra}{\rightarrow}

\newcommand{\RR}{\mathbb{R}}

\newtheorem{theorem}{Theorem}

\newtheorem{lemma}{Lemma}
\newtheorem{proposition}[theorem]{Proposition}
\newtheorem{definition}{Definition}

\definecolor{bluegreen}{rgb}{0.0, 0.5, 0.5}
\definecolor{redbrown}{rgb}{0.6, 0.2, 0.2}
\definecolor{bluebrown}{rgb}{0.4, 0.3, 0.4}
\definecolor{redgreen}{rgb}{0.5, 0.5, 0.0}

\newcolumntype{L}[1]{>{\raggedright\let\newline\\\arraybackslash\hspace{0pt}}m{#1}}
\newcolumntype{C}[1]{>{\centering\let\newline\\\arraybackslash\hspace{0pt}}m{#1}}
\newcolumntype{R}[1]{>{\raggedleft\let\newline\\\arraybackslash\hspace{0pt}}m{#1}}

\begin{document}
\onehalfspacing
\begin{titlepage}
\title{Selling Supplemental Information \thanks{I am indebted to Teddy Mekonnen for his continuous guidance, comments, and feedback. I also want to thank Jack Fanning for his comments and suggestions that improved this paper. I am grateful to Alessandro Bonatti, Geoffroy De Clippel,  Klajdi Hoxha, Bobby Pakzad-Hurson, Roberto Serrano, Rajiv Vohra, and Rakesh Vohra for their helpful comments. Finally, I thank the participants at the Brown University Theory Seminar, Bilkent University Microeconomics Seminar, and the Stony Brook Game Theory Conference. All errors are my own.}}
\author{Arlindo Skënderaj \thanks{Brown University, Department of Economics, \href{mailto:arlindo_skenderaj@brown.edu}{\texttt{arlindo\_skenderaj@brown.edu}}.} }
\date{\today}
\end{titlepage}
\maketitle
\begin{center} \href{https://drive.google.com/file/d/1RYgc_EHdTqa6Fy76EVMOmXIpaWEGz7tH/view?usp=share_link}
     {Please click here for the most recent version.}
\end{center}
\bigskip
\begin{abstract}
\noindent
I consider an environment in which a decision maker faces uncertainty and privately owns a signal about the true state of the world. Before observing the signal realization, the decision maker purchases additional information from a data broker. I characterize the data broker’s optimal selling mechanism, which involves screening over all signals. In the binary-action setting, the data broker extracts the efficient surplus by offering a distinct binary signal for each type. In more general environments, I provide conditions on the payoff structure under which efficient surplus extraction is achievable and give an explicit construction of the optimal menu. I also identify environments in which efficient surplus extraction fails and characterize the optimal menu when there are two decision maker types.
 \\

\vspace{0in}
\noindent\textbf{Keywords:} Data, Information Design, Mechanism Design, Multidimensional Screening, Selling Information  \\
\vspace{0in}
\end{abstract}


\newpage
\section{Introduction}
Decision makers are increasingly relying on third-party information to make better decisions. This information is usually sold in the form of datasets, action recommendations, or consulting services. While firms often collect first-party data through direct interactions or internal systems, this information is typically incomplete or noisy. For this reason, they rely on data brokers such as Acxiom, Experian, or Nielsen to obtain additional data. Other decision makers have limited knowledge and therefore may acquire additional information from experts. \\

The increasing demand for information, particularly data, has led to the widespread use of “data append” services, which allow decision makers to enrich existing datasets with additional attributes. For example, in the context of consumer data, the U.S. Federal Trade Commission (2014) reports that data brokers sell information by appending firms’ existing customer records with missing or additional details, such as age, gender, household income, or shopping behavior. These data are then used for consumer segmentation, targeted advertising, and various forms of price discrimination. Importantly, supplemental information does not replace existing information; instead, it is explicitly designed to be combined with it. Decision makers already have a plan based on the information available to them but may rely on experts or intermediaries to make more informed decisions. Examples include a lender evaluating loan applications after preliminary screening, or an employer deciding whether to hire a candidate after observing limited credentials.\\

This paper studies a data broker who sells information to a decision maker. The decision maker faces uncertainty about the state of the world and, consequently, about which action to take. She privately holds partial information in the form of a signal à la \citet{GentzkowKamenica2017}; this signal defines her type. The decision maker acquires additional information from the data broker to improve her decision. The general problem faced by the data broker consists of screening over all signals, which can be challenging for two main reasons. First, this is a multidimensional screening problem. Second, the type space is not totally ordered due to both horizontal and vertical differentiation across signals. The value of any given signal depends on the signal the decision maker already possesses; a particular signal might be highly valuable to one type but entirely uninformative to another. Much of the existing mechanism design literature relies on ranked type spaces where the single-crossing property holds. In this setting, it is unclear whether single-crossing conditions hold in general. Moreover, the decision maker's  type changes as she acquires additional data, as her payoff is derived from a utility-maximization problem.\\

I first analyze a setting in which the decision maker chooses between two actions. I provide an explicit construction of the optimal menu and show that the data broker extracts the efficient surplus: the menu induces each type to take the optimal action, and no type receives positive rent. A key feature of this menu is the \emph{minimal complementary signal}. Given two types $\pi$ and $\pi'$, the $\pi'$-minimal complementary signal to $\pi$ is the least valuable (under the payoff order for $\pi'$) signal that, when combined with $\pi$, induces type $\pi$ to take the optimal action. In the binary-action case, the $\pi'$-minimal complementary signal of $\pi$ depends only on the base signal $\pi$ and consists of only two messages. It can also be interpreted as an action recommendation that advises switching actions whenever the interim action induced by $\pi$ is not optimal, and not switching otherwise. An important property of this environment is that the incremental value of $\pi^c$ as a supplemental signal is maximized when the base signal is $\pi$. By mimicking another type and purchasing that type’s recommendation, the decision maker can choose either to obey or to disobey the recommendation. If she obeys, she switches actions whenever the other type switches. The other type always benefits from switching, whereas this type sometimes benefits and sometimes overturns an optimal action. Furthermore, for any given state, the marginal benefit of being corrected is the same for both types. If she disobeys, the outcome is equivalent to first switching the default action and then obeying. Since switching the default imposes a loss, and the benefit from obeying cannot exceed that of the informed type, disobeying is never better.
 \\

The efficient surplus extraction result does not extend to settings with many actions and states. When moving from two to many actions, the marginal benefit of being corrected generally depends on both the action and the state. As a result, for a given state, the marginal benefit of being corrected in that state differs across types of decision makers. This prevents the data broker from designing a menu that extracts the efficient surplus, as in the binary-action case. I construct an example at the end of Section~4 that illustrates this failure.\\

Nevertheless, efficient surplus extraction remains possible in a large class of cases. Proposition~\ref{generalproposition} provides sufficient conditions on the decision maker’s payoff function for efficient surplus extraction and gives an explicit construction of the optimal menu. The menu assigns each decision maker an \emph{error-correcting refining} signal that sends (i) a \emph{correction} message revealing the state whenever the interim action is suboptimal, and (ii) a \emph{confirmation} message whenever the interim action is optimal.
\\

The choice of modeling information as a signal à la \cite{GentzkowKamenica2017} is important for this environment. The decision maker augments her private information with supplemental information purchased from the data broker. Hence, it is crucial for the information to be defined in a way that equips it with a clear algebraic structure to combine distinct information structures. Several studies on the optimal sale of information represent information through Blackwell experiments (e.g., \cite{bergemann18}) or market segmentation (e.g., \cite{Yang2022}). However, these approaches do not provide an algebraic structure over the set of all information structures and are therefore not suitable for the environment considered in this paper. In contrast, representing information as a signal allows for a well-defined operation of combining two signals and analyzing their joining informational content. Moreover, it also captures how observations from different sources of information are correlated.

\subsection{Illustrative examples}
\subsubsection{Example: Monopolist buying consumer data.} 

A monopolist sells concert tickets to two consumer groups: students and professionals, with a willingness to pay of \$50 and \$100, respectively. The aggregate market consists of 30\% students and 70\% professionals. Suppose the monopolist owns a residential dataset that indicates whether a consumer lives in a house or an apartment. This information allows the monopolist to segment the market into two groups: one consisting mostly of professionals and one that includes a mix of both types. Specifically, 65\% of the market lives in houses, of which 14\% are students and 86\% are professionals; the remaining 35\% live in apartments, with 60\% students and 40\% professionals. Using this segmentation, the monopolist achieves a profit of  
\[
\frac{65}{100} \cdot 86 + \frac{35}{100} \cdot 50 = 73.4.
\]
In contrast, if the monopolist had no access to any dataset, she would set a uniform price of \$100 and sell only to professionals, resulting in a profit of \$70. If the monopolist had full information, she would be able to perfectly price discriminate and extract the entire surplus from the market, yielding a profit of
\[
50 \cdot 0.3 + 100 \cdot 0.7 = 85.
\]

Now suppose a data broker can sell additional information to the monopolist. The data broker does not know whether the monopolist already owns the residential dataset but believes she does with probability \( \frac{1}{2} \). To maximize revenue, the broker offers a menu of datasets:
\[
\mathcal{M} = \{(\text{fully-revealing data, \$15}), (\text{civil-status data, \$11.6})\}.
\]
Assume that only students who live in houses and professionals who live in apartments are married. Under this assumption, the civil-status dataset complements the residential dataset in such a way that their combination reveals the consumer type perfectly.\\

A monopolist with no prior data purchases the fully-revealing dataset and pays \$15, which is her willingness to pay for full information. Her willingness to pay for the civil-status data alone is zero. A monopolist who already owns the residential dataset purchases the civil-status data instead, because together they yield full segmentation, and her willingness to pay for this supplemental information is exactly \$11.6. In this way, the data broker successfully screens both types and extracts all information rents from each.

\subsubsection{Clinical diagnosis of a patient.} A physician faces three conditions $\Omega=\{\omega_1,\omega_2,\omega_3\}$ (bacterial pneumonia, viral pneumonia, lung cancer), each with prior probability $1/3$. There are three treatments $A=\{a_1,a_2,a_3\}$ (antibiotics, antivirals, chemotherapy). Payoffs are diagonal: 
\[
u(a_i,\omega_j)=
\begin{cases}
10, & i=j,\\
0, & i\neq j.
\end{cases}
\]
A bacterial culture test perfectly identifies bacterial pneumonia, whereas imaging does not distinguish viral pneumonia from early-stage lung cancer (both can produce similar opacities). The resulting partition is $\{\{\omega_1\},\{\omega_2,\omega_3\}\}$. When the culture is positive, the physician prescribes $a_1$ (antibiotics). When imaging is ambiguous, she chooses one treatment—say $a_2$ (antivirals)—which is correct in $\omega_2$ and incorrect in $\omega_3$. The resulting expected utility is $10/3+10/3+0=20/3$. Without any additional information, prescribing $a_1$ yields $10/3$.\\

Suppose a private laboratory can sell any diagnostic test but is uncertain whether the physician can run the culture test in house. The lab offers a menu of recommendation signals: conditional on the physician’s tentative treatment, the signal either (i) confirms it is correct or (ii) indicates it is incorrect and reveals the true state. Price the “informed-type” recommendation at $10/3$ and the “uninformed-type” recommendation at $20/3$.\\

Incentives are as follows. For the informed type (with partition $\{\{\omega_1\},\{\omega_2,\omega_3\}\}$), the recommendation delivers full information and raises expected utility from $20/3$ to $10$, so the incremental value is $10/3$; paying $10/3$ leaves zero rent. Pruchasing the uninformed type's recommendation, priced at $20/3$ is unprofitable.  For the uninformed type (who would otherwise prescribe $a_1$), the recommendation also yields $10$, i.e., an incremental value of $20/3$; paying $20/3$ leaves zero rent. Moreover, the uninformed type does not gain by imitating the informed type: that signal corrects only the mistakes defined relative to the informed type’s baseline (it reveals in $\omega_3$ but not in $\omega_2$), yielding at most the same incremental value $10/3$. 
\\

When can the data broker extract the efficient surplus?
With two actions, efficient surplus extraction is feasible for any belief about the decision maker’s private information. Nevertheless, this result generally does not extend to environments with multiple actions and states. Section 4 provides sufficient conditions under which efficient surplus extraction is achievable in these cases.\\

The structure of the paper is as follows: In Section 2, I provide related literature.  In Section 3, I introduce the model, and in Section 4, I conduct the main analysis. The conclusion is presented in Section 5. Proofs are provided in the Appendix.

\section{Related Literature} 
This paper relates to the literature on information design, and in particular to the optimal sale of information. Seminal work includes \cite{admati1986}, \cite{admati1990}, and \cite{allen1990}, which examine how a data broker can optimally design information structures for imperfectly informed decision makers. For an overview of markets for information, see \cite{BergemannBonatti2019}. Recent work, such as \cite{Babaioff,bergemann18,liu2021optimal,Yang2022,segura-rodriguez22,hoxha2024selling}, study environments in which a data broker sells information to an agent facing a decision problem under uncertainty. In \cite{Babaioff}, the data broker sells information to a decision maker who has a private payoff function. The information is sold after the data broker observes the realized state of the world. In \cite{liu2021optimal}, the data broker sells information to a decision maker who takes passive and active action. In their case, the payoff depends on the state only under the active action, reducing the problem to one dimension.  \cite{Yang2022} studies the optimal selling mechanism of a data broker who sells a signal to a monopolist with a privately known constant marginal cost.  \cite{segura-rodriguez22} studies a data broker who sells information about consumers to a monopolist who wants to predict one characteristic; which characteristic she wants to predict and how much she values the information are both private. \cite{hoxha2024selling} considers an environment in which the data broker sells information at the interim stage to a decision maker facing correlated states. This paper is closest to \cite{bergemann18}, who analyze a setting in which the broker sells Blackwell experiments to a decision maker whose private information is an interim belief about the true state; the signal is used to improve decision choices across states. In contrast, I study an environment in which the decision maker’s private information is the signal she owns, which represents a distribution over her interim beliefs. This captures situations where the decision maker must make frequent decisions and purchases information before the interim belief is realized. The key difference is that the monopolist’s private information is itself a signal structure, and the analysis focuses on the ex-ante value of refining that information. \\

On the literature of persuasion without transfers, \cite{kolotilin17}, \cite{kolotilin2018optimal}, \cite{krahmer20}, \cite{candogan2023optimal}, \cite{yamashita21}, \cite{ichihashi2025data}, and \cite{zhu23} analyze optimal mechanisms without transfers for providing information to privately informed agents.  Other studies have analyzed information transmission from experts to customers (see \cite{priceofadvice}, \cite{milgrom1986}, \cite{pitchik1987honesty}). \cite{esHo2007optimal} analyze how a seller should reveal information to bidders in an auction and show that a handicap auction can implement the revenue-maximizing information disclosure policy. \cite{bergemann2015} studies how a data broker sells consumer tracking information (cookies) to advertisers, and how the sale of that information affects advertising strategies and market  outcomes. In their paper, the price of the information was determined in a competitive market, and the data buyer did not have any private information. \cite{bergemann15} analyze the welfare outcomes that a designer can implement by offering any signal structure. They show that when the designer knows the monopolist's type, any division of the total surplus between the buyer and the seller can be achieved, provided that the total surplus does not exceed the efficient surplus, the monopolist’s profit is at least as much as the profit without information, and the buyer’s surplus is non-negative. \cite{krahmer20} studies mechanism design problems with quasi-linear utility, where the principal can design and disclose additional information that affects agents' valuations. They show that the principal can design information and a mechanism to fully extract the efficient surplus in a large class of cases. However, in these papers the agent does not face a decision problem under uncertainty. \\

This paper also relates to the literature on quality degradation (see \cite{Mussa1978} and \cite{Maskin1984}). \cite{Mussa1978} analyze how a monopolist can use product quality differentiation to maximize profits by setting different qualities and prices for different consumer segments. \cite{Maskin1984} study how a monopolist can design optimal incentive schemes under asymmetric information by offering quantity discounts tailored to consumers with private information about their preferences. In the current paper, the data seller's optimal mechanism may involve degrading the quality of information for some buyers to charge higher prices to others.\\

Methodologically, this paper contributes to the literature on multidimensional screening (\cite{rochet2003economics}, \cite{haghpanah2025screening}, and \cite{daskalakis2015strong}). Unlike standard screening problems, the single-crossing property does not hold in this setting, which makes the problem more challenging to analyze using existing tools. Moreover, the decision maker’s type changes as she acquires additional information, as her type is  endogenously determined through a utility maximization problem.


\section{Model} \label{sec:data}
A decision maker \emph{(she)} faces uncertainty about the state of the world $\omega$, which is drawn from a finite set $\Omega = \{\omega_1,\omega_2,...,\omega_n\}$. The decision maker chooses an action $a$ from the set of actions $A=\{a_1,a_2,...,a_m\}$. The payoff function given the action and the state is
\begin{equation*}
    u: \, A \times \Omega \rightarrow \RR.
\end{equation*}
The payoff function can also be represented by the following matrix:
\[
\begin{array}{c|cccc}
u(a_i, \omega_j) & \omega_1 & \omega_2 & \cdots & \omega_n \\ \hline
a_1 & u_{11} & u_{12} & \cdots & u_{1n} \\
a_2 & u_{21} & u_{22} & \cdots & u_{2n} \\
\vdots & \vdots & \vdots & \ddots & \vdots \\
a_m & u_{m1} & u_{m2} & \cdots & u_{mn}
\end{array}
\]
Assume that there is a unique optimal action in each state. That is,
\begin{equation*}
    \forall \omega \in \Omega, \exists a^*(\omega) \in A, \, u(a^*(\omega),\omega) > u(a,\omega), \, \forall a\in A, a\neq a^*(\omega). 
\end{equation*}
Let $\mu_0 \in \Delta \Omega$ be the common prior on $\Omega$. For any belief $\mu \in \Delta(\Omega)$, the decision maker takes an optimal action $a(\mu)$ that solves
\begin{equation*}
   \max_{a \in A} \sum_{\omega \in \Omega} \mu(\omega)\, u(a,\omega).
\end{equation*}
\subsection{Information}
Besides the prior, the decision maker privately holds additional information that enables her to make better decisions. I model the information as a signal $\pi$, following \cite{GentzkowKamenica2017}: a \textit{signal} is a finite partition of the extended state space $\Omega \times X$, where $X$ is a random variable, independent of $\omega$ and uniformly distributed on $[0,1]$. Let $S$ denote the collection of nonempty Lebesgue-measurable subsets of $\Omega \times X$, and write $\pi \subset S$.  An element $s \in \pi$ is called a \textit{message} (or a \textit{signal realization} ).
Define
\begin{equation*}
    \mu(s \mid \omega) = \lambda\{x \mid (\omega,x)\in s\} \quad \text{ and } \quad  \mu(s) = \sum_{\omega \in \Omega} \mu(s \mid \omega)\,\mu(\omega),
\end{equation*}
where $\lambda(\cdot)$ denotes Lebesgue measure. For any $\omega \in \Omega$ define 
\begin{equation*}
    s(\omega) = \{ x\in [0,1] \mid (\omega,x) \in s\}. 
\end{equation*}
Moreover, $\mu(s \mid \omega)$ denotes the conditional probability of receiving $s$ in state $\omega$, and $\mu(s)$ denotes the unconditional probability of receiving $s$.  For each $\omega \in \Omega$, the sets $\{s(\omega)\}_{s \in \pi}$ form a partition of $[0,1]$. That is,
\begin{equation}\label{partition}
    \sum_{s \in \pi} \mu(s \mid \omega) = 1 \quad \forall \, \omega \in \Omega.
\end{equation}
Let $\Pi$ denote the set of all signals, and let $\overline{\pi} =\{(\omega,[0,1]) \mid \omega \in \Omega\}$ denote the fully revealing signal. Modeling information as in \cite{GentzkowKamenica2017} provides an algebraic structure over $\Pi$ that allows us to “add” signals and analyze their joint information content. \footnote{For more on the comparison of signals \cite{Brooks2024}.} It also captures how observations from one source correlate with observations from others. \footnote{Several studies on selling information represent information as Blackwell experiments \cite{bergemann18} or as market segmentations \cite{Yang2022}. These representations lack a convenient algebraic structure over the space of information structures and are therefore less suitable for the environment of this paper.}

\begin{figure}
\begin{center}
\scalebox{0.9}{%
\begin{tabular}{c c}
    $\omega_1$ & $\omega_2$ \\
    \begin{tikzpicture}
        \node[left] at (0,0) {$\pi:$\quad};
        \draw[thick] (0,0) -- (4,0);
        \draw[blue, thick] (0,0) -- (3,0);
        \draw[red, thick] (3,0) -- (4,0);
        \draw[blue]
            (0,0) -- (3,0) node[midway, yshift=-0.5cm] {\textcolor{blue}{$s_1$}};
        \draw[red]
            (3,0) -- (4,0) node[midway, yshift=-0.5cm] {\textcolor{red}{$s_2$}};
        \filldraw[black] (0,0) circle (2pt) node[below] {$0$};
        \filldraw[black] (3,0) circle (2pt);
        \filldraw[black] (4,0) circle (2pt) node[below] {$1$};
    \end{tikzpicture} &
    \begin{tikzpicture}
        \draw[thick] (0,0) -- (4,0);
        \draw[blue, thick] (0,0) -- (1,0);
        \draw[red, thick] (1,0) -- (4,0);
        \draw[blue]
            (0,0) -- (1,0) node[midway, yshift=-0.5cm] {\textcolor{blue}{$s_1$}};
        \draw[red]
            (1,0) -- (4,0) node[midway, yshift=-0.5cm] {\textcolor{red}{$s_2$}};
        \filldraw[black] (0,0) circle (2pt) node[below] {$0$};
        \filldraw[black] (1,0) circle (2pt);
        \filldraw[black] (4,0) circle (2pt) node[below] {$1$};
    \end{tikzpicture} \\
\end{tabular}
}
\end{center}
\caption{A signal.}
\label{signal}
\end{figure}
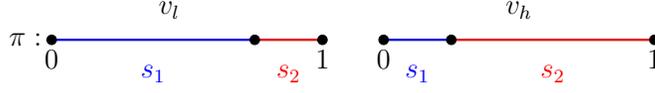
\subsubsection{Value of supplemental information}
If the decision maker owns a signal $\pi$, then for each message $s\in \pi$, she will update her prior belief $\mu_0$ and choose the optimal action $a(s)$ that solves 
\begin{equation}
     \max_{a \in A} \sum_{\omega \in \Omega} \mu(\omega \mid s)\, u(a,\omega).
\end{equation}
Her interim utility conditional on receiving a message $s$ is
\begin{equation}
    u(s) = \sum_{\omega \in \Omega} \mu(\omega \mid s) u(a(s), \omega).
\end{equation}
Her expected utility from a signal $\pi$ is therefore given by 
\begin{equation}
    U(\pi) = \sum_{s\in \pi}  \mu(s) \, u(s),
\end{equation}
Her expected utility from always taking the optimal action is \footnote{Note that the decision maker does not necessarily need a fully revealing signal to always take the optimal action. To see this, consider three states and two actions, where the first action is optimal in the first two states and the second action is optimal in the third. A signal that separates state 3 from the first two states does not fully reveal the state, yet it still induces the optimal action in every state.}
\begin{equation}
   \overline{U} = \sum_{\omega \in \Omega} \mu(\omega) \, u(a^*(\omega), \, \omega). 
\end{equation}
Next, given any two signals $\pi,\pi' \in \Pi$, define their \emph{join} $\pi \lor \pi'$ by
\begin{equation}
    \pi \lor \pi' = \{s \cap s' \mid s\in \pi, s' \in \pi'\}. 
\end{equation}
That is, $\pi \lor \pi'$ is the signal that gives the same information as observing the two signals $\pi$ and $\pi'$. The join to two signals is illustrated in Figure~\ref{signaljoin}. The expected utility of $\pi \lor \pi'$ is
\begin{equation}
    U(\pi \lor \pi') = \sum_{s\in \pi} \sum_{s' \in \pi'} \mu(s\cap s') u(s \cap s'). 
\end{equation}
The net value from adding $\pi'$ to $\pi$ is
\begin{equation}
    V(\pi' \mid \pi) = U(\pi' \lor \pi) - U(\pi). 
\end{equation}

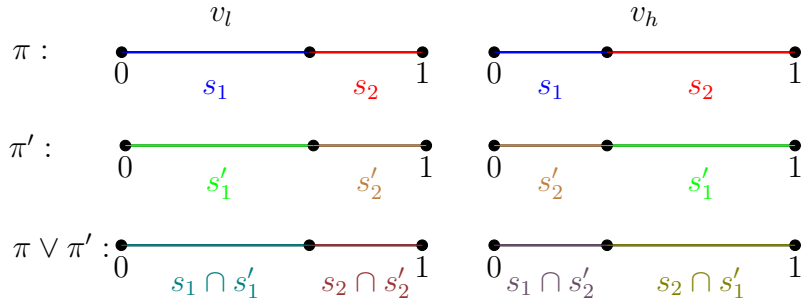
\begin{figure}[H]
\centering
\scalebox{1}{%
\begin{tabular}{c c}
    $\omega_1$ & $\omega_2$ \\
    \begin{tikzpicture}[scale=1]
        \node[left] at (0,0) {$\pi:$\qquad \qquad};
        \draw[thick] (0,0) -- (4,0);
        \draw[blue, thick] (0,0) -- (2.5,0);
        \draw[red, thick] (2.5,0) -- (4,0);
        \draw[blue]
        (0,0) -- (2.5,0) node[midway, yshift=-0.5cm] {$s_1$};
        \draw[red]
        (2.5,0) -- (4,0) node[midway, yshift=-0.5cm] {$s_2$};
                \filldraw[black] (2.5,0) circle (2pt);
        \filldraw[black] (0,0) circle (2pt) node[below] {$0$};
        \filldraw[black] (4,0) circle (2pt) node[below] {$1$};
    \end{tikzpicture} &
    \begin{tikzpicture}[scale=1]
        \draw[thick] (0,0) -- (4,0);
        \draw[blue, thick] (0,0) -- (1.5,0);
        \draw[red, thick] (1.5,0) -- (4,0);
        \draw[blue]
        (0,0) -- (1.5,0) node[midway, yshift=-0.5cm] {$s_1$};
        \draw[red]
        (1.5,0) -- (4,0) node[midway, yshift=-0.5cm] {$s_2$};
        \filldraw[black] (1.5,0) circle (2pt);
        \filldraw[black] (0,0) circle (2pt) node[below] {$0$};
        \filldraw[black] (4,0) circle (2pt) node[below] {$1$};
    \end{tikzpicture} \\
    \begin{tikzpicture}[scale=1]
        \node[left] at (0,0) {$\pi':$\qquad \qquad};
        \draw[thick] (0,0) -- (4,0);
        \draw[green]
        (0,0) -- (1.5,0) node[midway, yshift=-0.5cm] {$s_1'$};
        \draw[brown]
        (1.5,0) -- (4,0) node[midway, yshift=-0.5cm] {$s_2'$};
        \filldraw[black] (0,0) circle (2pt) node[below] {$0$};
        \filldraw[black] (4,0) circle (2pt) node[below] {$1$};
        \filldraw[black] (1.5,0) circle (2pt);
    \end{tikzpicture} &
    \begin{tikzpicture}[scale=1]
        \draw[thick] (0,0) -- (4,0);
        \draw[brown, thick] (0,0) -- (3,0);
        \draw[green, thick] (3,0) -- (4,0);
        \draw[brown]
        (0,0) -- (3,0) node[midway, yshift=-0.5cm] {$s_2'$};
        \draw[green]
        (3,0) -- (4,0) node[midway, yshift=-0.5cm] {$s_1'$};
        \filldraw[black] (3,0) circle (2pt);
        \filldraw[black] (0,0) circle (2pt) node[below] {$0$};
        \filldraw[black] (4,0) circle (2pt) node[below] {$1$};
    \end{tikzpicture} \\
    \begin{tikzpicture}[scale=1]
        \node[left] at (0,0) {$\pi\lor\pi':$};
        \draw[thick] (0,0) -- (4,0);
        \draw[bluegreen, thick] (0,0) -- (1.5,0);
        \draw[redbrown, thick] (2.5,0) -- (4,0);
        \draw[bluegreen]
        (0,0) -- (1.5,0) node[midway, yshift=-0.5cm] {$s_1 \cap s_1'$};
        \draw[bluebrown]
        (1.5,0) -- (2.5,0) node[midway, yshift=-0.5cm] {$s_1\cap s_2'$};
        \draw[redbrown]
        (2.5,0) -- (4,0) node[midway, yshift=-0.5cm] {$s_2 \cap s_2'$};
        \filldraw[black] (1.5,0) circle (2pt) ;
        \filldraw[black] (2.5,0) circle (2pt) ;
        \filldraw[black] (0,0) circle (2pt) node[below] {$0$};
        \filldraw[black] (4,0) circle (2pt) node[below] {$1$};
    \end{tikzpicture} &
    \begin{tikzpicture}[scale=1]
        \draw[thick] (0,0) -- (4,0);
        \draw[bluebrown, thick] (0,0) -- (1.5,0);
        \draw[redgreen, thick] (1.5,0) -- (4,0);
        \draw[bluebrown]
        (0,0) -- (1.5,0) node[midway, yshift=-0.5cm] {$s_1 \cap s_2'$};
        \draw[redbrown]
        (1.5,0) -- (3,0) node[midway, yshift=-0.5cm] {$s_2 \cap s_2'$};
        \draw[redgreen]
        (3,0) -- (4,0) node[midway, yshift=-0.5cm] {$s_2 \cap s_1'$};
        \filldraw[black] (1.5,0) circle (2pt);
        \filldraw[black] (3,0) circle (2pt);
        \filldraw[black] (0,0) circle (2pt) node[below] {$0$};
        \filldraw[black] (4,0) circle (2pt) node[below] {$1$};
    \end{tikzpicture} \\
\end{tabular}
}
\caption{The join of two signals.}
\label{signaljoin}
\end{figure}

\subsection{Data broker's problem} 
The data broker screens the decision maker based on the signal she already has. Let $\mathcal{T}\subseteq\Pi$ denote the decision maker’s type space, with types distributed according to $F$ on $\mathcal{T}$. By the revelation principle the data broker can restrict attention to direct mechanism which offers a signal $\sigma(\pi)$ at a price $t(\pi)$ to each type $\pi \in \mathcal{T}$.  Timing is as follows:
\begin{itemize}
  \item[(i)] The data broker posts a menu $\mathcal{M}=\{(\sigma(\pi),\,t(\pi))\}_{\pi \in \mathcal{T}}$ of signal–price pairs before the true state is realized.
  \item[(ii)] A decision maker of type $\pi$ purchases a signal $\sigma(\pi)$ and pays $t(\pi)$, before her private message is realized.
  \item[(iii)] The decision maker observes a signal $s$ from $\pi$ and another signal $s'$ from $\sigma(\pi)$ and chooses an action $a \in A$. 
\end{itemize}
Define
\begin{equation*}
    V(\pi) \triangleq  V(\sigma(\pi) \mid \pi) = U(\sigma(\pi) \lor \pi) - U(\pi),
\end{equation*}
to be the net benefit of type $\pi$ from reporting truthfully. The data broker chooses a menu to maximize expected revenue
\begin{equation}\label{eq:P}
    \int_{\mathcal{T}} t(\pi)\,\mathrm{d}F(\pi),
\end{equation}
subject to incentive compatibility constraints
\begin{equation*}\tag{IC}\label{IC}
    V(\pi) - t(\pi) \, \ge\, V(\sigma(\pi')\mid \pi) - t(\pi'), \quad \forall\, \pi,\pi' \in \mathcal{T},
\end{equation*}
and participation constraints
\begin{equation*}\tag{IR}\label{IR}
     V(\pi) - t(\pi) \, \ge\, 0, \quad \forall\, \pi \in \mathcal{T}.
\end{equation*}

\section{Analysis}
This section characterizes the optimal menu for any given distribution over the type space. In a binary action setting, the data broker extracts the first best surplus. Each type receives a distinct signal with binary messages. In environments with many actions and states, efficient surplus extraction generally fails. I provide conditions on the payoff structures under which the efficient surplus extraction is implementable. I define the efficient surplus extraction below. 

\begin{definition}
The data broker extracts the efficient surplus from the decision maker if there exists a feasible menu
\[
\mathcal{M}^* = \{(\sigma^*(\pi), t^*(\pi))\}_{\pi \in \mathcal{T}}
\]
such that, for every $\pi \in \mathcal{T}$, the induced posterior satisfies $U(\sigma^*(\pi) \lor \pi) = \overline{U}$,  and the transfer satisfies $t^*(\pi) = \overline{U} - U(\pi)$.
\end{definition}

\subsection{Two types}
Consider $\mathcal{T} = \{\pi_L, \pi_H\},$ with $U(\pi_L) > U(\pi_H)$. Interpret the high type as the less informed type, i.e., the one with the higher willingness to pay for full information. Data broker offers a menu $(\sigma_L,t_L),(\sigma_H,t_H)$ that solves
\begin{equation}\tag{$\mathcal{P}$}\label{eq:P}
\begin{aligned}
    & \underset{((\sigma_H,t_H),(\sigma_L,t_L))}{\text{max}} \rho \ t_H+ (1-\rho) \ t_L\\
    & \text{subject to}\\
    & (IC_H): \quad V_H(\sigma_H) - t_H \geq V_H(\sigma_L) - t_L, \\
    & (IC_L): \quad V_L(\sigma_L) - t_L \geq V_L(\sigma_H) - t_H, \\
    &(IR_H): \quad V_H(\sigma_H) - t_H \geq 0,  \\
    &(IR_L): \quad V_L(\sigma_L) - t_L \geq 0, 
\end{aligned}
\end{equation}
where $V_i (\sigma_j) = V(\sigma_j \mid \pi_i)$ for $i,j \in \{L,H\}$.
\begin{lemma}\label{reduceproblem}
The optimal menu $(\sigma_L,t_L),(\sigma_H,t_H)$ is such that 
\begin{itemize}
    \item [a.] The high type always take the optimal action i.e., $U(\sigma_H \lor \pi_H) = \overline{U}$.
    \item [b.] $IR_L$ binds. 
    \item [c.] If the efficient surplus extraction is not possible, then $IC_H$ binds. 
\end{itemize}
\end{lemma}
The proof is provided in Appendix~\ref{appendix:reduceproblem} and it follows an argument similar to that in \cite{haghpanah2025screening}. Lemma~\ref{reduceproblem} reduces the problem to 
\begin{equation*}
\begin{aligned}
    & \underset{((\sigma_H,t_H),(\sigma_L,t_L))}{\text{max}} \rho \ t_H+ (1-\rho) \ V_L(\sigma_L)\\
    & \text{subject to}\\
    & (IC_H): \quad V_H(\overline{\pi}) - t_H \geq V_H(\sigma_L) - V_L(\sigma_L), \\
    &(IR_H): \quad V_H(\overline{\pi}) - t_H \geq 0,  \\
    & (IC_L): \quad  V_L(\overline{\pi}) - t_H \leq 0. 
\end{aligned}
\end{equation*}
From Lemma~\ref{reduceproblem}(c), if the efficient surplus extraction is not possible, the data broker sets $t_H = V_H(\overline{\pi}) - V_H(\sigma_L) + V_L(\sigma_L)$. The problem reduces further to only choosing a $\sigma_L^*$ that solves
    \begin{equation}
    \max_{\sigma_L \in \Pi} \rho \, \big(V_H(\overline{\pi}) - V_H(\sigma_L) + V_L(\sigma_L)\big) +(1-\rho) \,V_L(\sigma_L).
\end{equation}
Since $V_i(\sigma_L) =  U(\sigma_L \lor \pi_i) - U(\pi_i)$ for $i \in \{L,H\}$, the problem is equivalent to 
\begin{equation}\label{reducedproblem}
    \max_{\sigma_L \in \Pi} U(\sigma_L \lor \pi_L) - \rho \, U(\sigma_L \lor \pi_H).
\end{equation}
The following proposition provides the sufficient and necessary conditions under which the efficient surplus extraction is possible. Before that consider the following definition of a minimal complementary signal. 
\begin{definition}
Given two signals $\pi,\pi' \in \Pi$, the \textbf{$\pi'$-minimal complementary signal of  $\pi$}, denoted by $\pi^c(\pi')$ is defined as 
\begin{equation}\label{quasiminimalcomplementarysignal}
\pi^c(\pi') \in \arg\min_{\tilde{\pi} \in \Pi} U(\pi' \lor \tilde{\pi})
\quad \text{s.t.}\quad
U(\tilde{\pi} \lor \pi) = \overline{U}.
\end{equation}
\end{definition}
\begin{proposition}
Let $\mathcal{T} = \{\pi_L,\pi_H\}$. Data broker extracts the efficient surplus if and only if 

\begin{equation*}\label{fullsurplusextraction}
    V_L(\pi^c_L(\pi_H)) \geq V_H(\pi^c_L(\pi_H)).
\end{equation*}
\end{proposition}
The proof is in the Appendix~\ref{appendix:fullsurplusextraction}.

\subsection{Binary actions}
Consider $A=\{a_1,a_2\}$, $\Omega = \{\omega_1,...,\omega_n\}$ and the utility function $u: A 
\times \Omega \ra \RR$ to be such that
\[
\begin{array}{c|cccc}
u(a,\omega) & \omega_1 & \omega_2 & \cdots & \omega_n \\ \hline
a_1 & u_{11} & u_{12} & \cdots & u_{1n} \\
a_2 & u_{21} & u_{22} & \cdots & u_{2n}
\end{array}
\]
Let the type space $\mathcal{T}$ follow any distribution $F$ over $\Pi$. 
For $\omega\in\Omega$, write $u_1(\omega)\triangleq u(a_1,\omega)$ and $u_2(\omega)\triangleq u(a_2,\omega)$. Since in each state there is a unique optimal action, we can partition the state space into 
\begin{equation*}
    \Omega_1=\{\omega\in\Omega:\; u_1(\omega) > u_2(\omega)\}
    \quad\text{and}\quad
    \Omega_2=\{\omega\in\Omega:\; u_2(\omega)>u_1(\omega)\}.
\end{equation*}
This aggregation reduces the problem to two regions in which each action is uniquely optimal, so the analysis follows the same logic as the binary state and binary action case.
\subsubsection{Minimal complementary signal.}
I begin by identifying the conditions under which efficient surplus extraction is possible. A key feature of the binary-action setting is that, for any two signals $\pi$ and $\pi'$, the $\pi'$-minimal complementary signal of $\pi$ is binary and independent of $\pi'$. The following lemma formalizes this observation.

\begin{lemma}\label{binary}
Let $|A| = 2$, and let $\pi \in \Pi$ be a non–fully revealing signal. For any $\pi' \in \Pi$, the $\pi'$-minimal complementary signal of $\pi$ is binary and does not depend on $\pi'$. I  denote this signal by $\pi^c$.
\end{lemma}

The proof is in the (Appendix \ref{appendix:binary}). With two actions, the minimal complementary signal has a simple interpretation: its messages can be interpreted as “don’t switch” when the interim action induced by $\pi$ is optimal, and “switch” otherwise. That is, a signal can also be interpreted as an action recommendation.  With two actions, the minimal complementary signal has a simple interpretation: its messages can be interpreted as “don’t switch” when the interim action induced by $\pi$ is optimal, and “switch” otherwise. That is, a signal can also be interpreted as an action recommendation. The following lemma characterizes an important property of the minimal complementary signal in the binary-action case. In particular, the incremental value of $\pi^c$ as a supplemental signal is maximized when the base signal is $\pi$.

\begin{figure}
\begin{center}
    \scalebox{1}{%
\begin{tabular}{c c c}
    $\omega_1$ & $\omega_2$ \\
   \begin{tikzpicture}
        \node[left] at (0,0) {$\pi:$\quad};
        \draw[thick] (0,0) -- (4,0);
        \draw[blue, thick] (0,0) -- (2.5,0);
        \draw[red, thick] (2.5,0) -- (4,0);
        \filldraw[black] (2.5,0) circle (2pt) node[below] {$a$};
        \filldraw[black] (0,0) circle (2pt) node[below] {$0$};
        \filldraw[black] (4,0) circle (2pt) node[below] {$1$};
        \draw[blue]
        (0.2,0) -- (2.4,0) node[midway, yshift=-0.5cm] {\textcolor{blue}{$s_1$}};
        \draw[red]
        (2.6,0) -- (3.9,0) node[midway, yshift=-0.5cm] {\textcolor{red}{$s_2$}};
    \end{tikzpicture} &
    \begin{tikzpicture}
        \draw[thick] (0,0) -- (4,0);
        \draw[blue, thick] (0,0) -- (1.5,0);
        \draw[red, thick] (1.5,0) -- (4,0);
        \filldraw[black] (1.5,0) circle (2pt) node[below] {$b$};
        \filldraw[black] (0,0) circle (2pt) node[below] {$0$};
        \filldraw[black] (4,0) circle (2pt) node[below] {$1$};
        \draw[blue]
        (0.2,0) -- (1.3,0) node[midway, yshift=-0.5cm] {\textcolor{blue}{$s_1$}};
        \draw[red]
        (1.7,0) -- (3.8,0) node[midway, yshift=-0.5cm] {\textcolor{red}{$s_2$}};
    \end{tikzpicture} \\
   \begin{tikzpicture}
        \node[left] at (0,0) {$\pi^c:$\quad};
        \draw[thick] (0,0) -- (4,0);
        \draw[purple, thick] (0,0) -- (2.5,0);
        \draw[cyan, thick] (2.5,0) -- (4,0);
        \filldraw[black] (2.5,0) circle (2pt) node[below] {$a$};
        \filldraw[black] (0,0) circle (2pt) node[below] {$0$};
        \filldraw[black] (4,0) circle (2pt) node[below] {$1$};
        \draw[purple]
        (0.2,0) -- (2.4,0) node[midway, yshift=-0.5cm] {\textcolor{purple}{$\hat{s}_1$}};
        \draw[cyan]
        (2.6,0) -- (3.9,0) node[midway, yshift=-0.5cm] {\textcolor{cyan}{$\hat{s}_2$}};
    \end{tikzpicture} &
   \begin{tikzpicture}
        \draw[thick] (0,0) -- (4,0);
        \draw[cyan, thick] (0,0) -- (1.5,0);
        \draw[purple, thick] (1.5,0) -- (4,0);
        \filldraw[black] (1.5,0) circle (2pt) node[below] {$b$};
        \filldraw[black] (0,0) circle (2pt) node[below] {$0$};
        \filldraw[black] (4,0) circle (2pt) node[below] {$1$};
        \draw[cyan]
        (0.2,0) -- (1.3,0) node[midway, yshift=-0.5cm] {\textcolor{cyan}{$\hat{s}_2$}};
        \draw[purple]
        (1.7,0) -- (3.8,0) node[midway, yshift=-0.5cm] {\textcolor{purple}{$\hat{s}_1$}};
    \end{tikzpicture} 
\end{tabular}
}
\end{center}
\caption{Minimal complementary signal with two actions.}
\end{figure}

\begin{lemma}\label{mainlemma}
Assume $|A|=2$. Let $\pi\in\Pi$ and $\pi^c$ be a minimal complementary signal to $\pi$. Then,
\[
    V(\pi^c \mid \pi) \, \ge\, V(\pi^c \mid \pi') \quad \forall \pi' \in \Pi.
\]
\end{lemma}

The proof is provided in Appendix \ref{appendix:mainlemma}. The lemma implies a key property of binary-action environments: a minimal complementary signal is (weakly) most valuable to the type for which it is designed. Below, I will provide the key intuition for this result using the following example. In the next section, this need not hold when there are more than three states.

\subsubsection{Intuition behind Lemma \ref{mainlemma}}
 Consider two states, good and bad, i.e., $\Omega=\{G,B\}$, and two actions $A=\{g,b\}$. Without loss of generality, normalize $u_{12}=u_{21}=0$, $u_{22}=1$, and $u_{11}=x>0$. The payoffs are given by the following matrix
\[
\begin{array}{c|cc}
   & g & b \\ \hline
G & x & 0 \\
B & 0 & 1
\end{array}
\]
Take any two types $\pi,\pi' \in \Pi$ and let $s_1$ be the union of all messages of $\pi$ that induce action $g$, and $s_2$ the union of all messages of $\pi$ that induce action $b$. Also let $\pi^c = \{\hat{s}_1,\hat{s}_2\}$ be the minimal complementary signal of $\pi$, as shown in Figure~\ref{example}. That is, $\pi^c$ sends message $\hat{s}_2$ whenever the message sent by $\pi$ induces a suboptimal action, and $\hat{s}_1$ whenever the message sent by $\pi$ induces the optimal action. We can think of the signal $\pi^c$ as an action recommendation whose messages have the following interpretation: $\hat{s}_1$ as ``keep the action induced by $\pi$'' and $\hat{s}_2$ as ``switch the action.''  The net benefit of adding $\pi^c$ to $\pi$ is given by 
\begin{equation*}
    V(\pi^c \mid \pi) = \mu(G, s_2) \cdot x + \mu(B, s_1) \cdot 1. 
\end{equation*}
For simplicity, assume that $\pi'$ consists of a single message $s'$. The same intuition extends to an arbitrary signal. Without loss of generality, assume that the optimal action under $s'$ is $b$. A decision maker who has $\pi'$ can mimic $\pi$ by either obeying or disobeying the recommendation of $\pi^c$. The net benefit from obeying the recommendation of $\pi^c$ is 
\begin{equation*}
    V(\pi^c \mid \pi') = \mu(G, s_2)\, x - \mu(B , s_1) \cdot 1.
\end{equation*}
This is clearly no greater than $V(\pi^c \mid \pi)$. When receiving the ``do not switch'' message, neither type gains any benefit. When receiving the ``switch'' message, type $\pi$ always improves her action, whereas type $\pi'$  takes the correct action with some probability and the wrong one with the remaining probability.
\\
On the other hand, type $\pi'$'s net benefit from disobeying the recommendation of $\pi^c$ is given by 
\begin{equation*}
    \begin{split}
        V(\pi^c \mid \pi') & = \mu(G \cap s_1) \, x - \mu(B \cap s_2) \\
        & = [\mu(G) \, x - \mu(B)] + [\mu(B \cap s_1) - \mu(G \cap s_2)\, x].
    \end{split}
    \end{equation*}
Disobeying the ``switch/don't switch'' recommendation is equivalent to first changing the default action from $b$ to $g$, and then obeying the recommendation as if $g$ were the default. First, switching the default action from $b$ to $g$ leads to a loss, since before receiving any information, taking action $b$ yields a higher expected payoff. Second, we have already shown that the benefit from obeying is no greater than $V(\pi^c \mid \pi)$. 
\begin{figure}
 \begin{center}
    \scalebox{1}{%
\begin{tabular}{c c}
    $G$ & $B$ \\
   \begin{tikzpicture}
        \node[left] at (0,0) {$\pi:$\quad};
        \draw[thick] (0,0) -- (4,0);
        \draw[blue, thick] (0,0) -- (2.5,0);
        \draw[red, thick] (2.5,0) -- (4,0);
        \filldraw[black] (2.5,0) circle (2pt) node[below] {$a$};
        \filldraw[black] (0,0) circle (2pt) node[below] {$0$};
        \filldraw[black] (4,0) circle (2pt) node[below] {$1$};
        \draw[blue]
        (0.2,0) -- (2.4,0) node[midway, yshift=-0.5cm] {\textcolor{blue}{$s_1$}};
        \draw[red]
        (2.6,0) -- (3.9,0) node[midway, yshift=-0.5cm] {\textcolor{red}{$s_2$}};
    \end{tikzpicture} &
    \begin{tikzpicture}
        \draw[thick] (0,0) -- (4,0);
        \draw[blue, thick] (0,0) -- (1.5,0);
        \draw[red, thick] (1.5,0) -- (4,0);
        \filldraw[black] (1.5,0) circle (2pt) node[below] {$b$};
        \filldraw[black] (0,0) circle (2pt) node[below] {$0$};
        \filldraw[black] (4,0) circle (2pt) node[below] {$1$};
        \draw[blue]
        (0.2,0) -- (1.3,0) node[midway, yshift=-0.5cm] {\textcolor{blue}{$s_1$}};
        \draw[red]
        (1.7,0) -- (3.8,0) node[midway, yshift=-0.5cm] {\textcolor{red}{$s_2$}};
    \end{tikzpicture} \\
    \begin{tikzpicture}
        \node[left] at (0,0) {$\pi^c:$\quad};
        \draw[thick] (0,0) -- (4,0);
        \draw[purple, thick] (0,0) -- (2.5,0);
        \draw[cyan, thick] (2.5,0) -- (4,0);
        \filldraw[black] (2.5,0) circle (2pt) node[below] {$a$};
        \filldraw[black] (0,0) circle (2pt) node[below] {$0$};
        \filldraw[black] (4,0) circle (2pt) node[below] {$1$};
        \draw[purple]
        (0.2,0) -- (2.4,0) node[midway, yshift=-0.5cm] {\textcolor{purple}{$\hat{s}_1$}};
        \draw[cyan]
        (2.6,0) -- (3.9,0) node[midway, yshift=-0.5cm] {\textcolor{cyan}{$\hat{s}_2$}};
    \end{tikzpicture} &
   \begin{tikzpicture}
        \draw[thick] (0,0) -- (4,0);
        \draw[cyan, thick] (0,0) -- (1.5,0);
        \draw[purple, thick] (1.5,0) -- (4,0);
        \filldraw[black] (1.5,0) circle (2pt) node[below] {$b$};
        \filldraw[black] (0,0) circle (2pt) node[below] {$0$};
        \filldraw[black] (4,0) circle (2pt) node[below] {$1$};
        \draw[cyan]
        (0.2,0) -- (1.3,0) node[midway, yshift=-0.5cm] {\textcolor{cyan}{$\hat{s}_2$}};
        \draw[purple]
        (1.7,0) -- (3.8,0) node[midway, yshift=-0.5cm] {\textcolor{purple}{$\hat{s}_1$}};
    \end{tikzpicture} \\
        \begin{tikzpicture}
        \node[left] at (0,0) {$\pi':$\quad};
        \draw[thick] (0,0) -- (4,0);
        \draw[brown, thick] (0,0) -- (4,0);
        \filldraw[black] (0,0) circle (2pt) node[below] {$0$};
        \filldraw[black] (4,0) circle (2pt) node[below] {$1$};
        \draw[brown]
        (0.2,0) -- (3.9,0) node[midway, yshift=-0.5cm] {\textcolor{brown}{$s'$}};
    \end{tikzpicture} &
   \begin{tikzpicture}
        \draw[thick] (0,0) -- (4,0);
        \draw[brown, thick] (0,0) -- (4,0);
        \filldraw[black] (0,0) circle (2pt) node[below] {$0$};
        \filldraw[black] (4,0) circle (2pt) node[below] {$1$};
        \draw[brown]
        (0.2,0) -- (3.9,0) node[midway, yshift=-0.5cm] {\textcolor{brown}{$s'$}};
    \end{tikzpicture} 
\end{tabular}
}
\end{center}
\caption{Intuition of Lemma~\ref{mainlemma}. }
\label{example}
\end{figure}
\begin{proposition}\label{mainproposition}
Assume $\Omega = \{\omega_1,...,\omega_n\}$, $A = \{a_1,a_2\}$, $\mathcal{T}=\Pi$ and $F\sim \mathcal{T}$.  The data broker extracts the efficient surplus from the decision maker. In particular, the optimal menu is given by
\begin{equation*}
    \mathcal{M} = \{(\pi^c, t({\pi})) \mid \pi \in \Pi\},
\end{equation*}
where 
\begin{equation*}
    t(\pi) = \overline{U} - U(\pi).
\end{equation*}

\end{proposition}
The proof is in Appendix \ref{appendix:mainproposition}.

\subsection{Many states and many actions}
The result of Proposition~\ref{mainproposition} does not extend beyond the binary action case. In this section, I analyze an environment with many states and actions and provide sufficient conditions on the payoff structure under which the data broker can still extract the efficient surplus. Assume $|A|=|\Omega|=n\ge 3$ and consider a class of payoffs where the decision maker’s objective is to match the action to the realized state. This class has two features: (i) any action--state mismatch yields the same baseline payoff; (ii) the marginal value of correction depends only on the state, not on the action the decision maker would otherwise take. Unlike general environments (where the marginal value depends on the chosen action), I restrict attention to payoff functions $u(\omega,a)$ that admit the following diagonal representation:

\[
\begin{array}{c|cccc}
u(a_i, \omega_j) & \omega_1 & \omega_2 & \cdots & \omega_n \\ \hline
a_1 & u_1 & 0 & \cdots & 0 \\
a_2 & 0 & u_2 & \cdots & 0 \\
\vdots & \vdots & \vdots & \ddots & \vdots \\
a_n & 0 & 0 & \cdots & u_n
\end{array}
\]
In other words,
\begin{equation}\label{diagonalmatrix}
u(\omega_i,a_j) =
\begin{cases}
  u_{i} & \text{if } j = i, \\
  0,  & \text{if } j \neq i.
\end{cases}
\end{equation}
Let $u(\omega)$ denote the payoff obtained from taking the optimal action in state $\omega$ i.e., 
\begin{equation*}
    u(\omega) =  \max_{a\in A} u(a,\omega).
\end{equation*}
In settings with many actions and states, for any two signals $\pi$ and $\pi'$, the $\pi'$-minimal complementary signal to $\pi$ is generally challenging to identify explicitly and, crucially, it depends on both $\pi$ and $\pi'$. As a result, constructing the optimal menu using minimal complementary signals, as in the binary-action case, does not extend to settings with many actions and states. In the next section, I construct an alternative signal that induces full information and depends only on the signal it augments.

\subsubsection{Error-correcting refinement}
For any $\pi \in \Pi$, construct $\hat{\pi}$ such that it reveals the true state whenever $\pi$ induces a suboptimal action and otherwise sends a common message confirming that the action is correct. Note that $\hat{\pi}$ depends only on $\pi$. I refer to $\hat{\pi}$ as an \textit{error-correcting refinement} of $\pi$. Formally, $\hat{\pi}$ is constructed as follows. First, define $s_i$'s such that $s_1$ is the union of all messages in $\pi$ that induce action $a_1$, and for any $i \in \{2,...,n\}$, $s_i$ is the union of all messages in $\pi$ that induce action $a_{i}$, and do not induce action $a_{k}$ for any $k < i$. \footnote{Note that a message may induce more than one action.} Next, construct $\hat{\pi} = \{\hat{s}_1,\hat{s}_2,...,\hat{s}_n,\hat{s}\}$ from $\{s_i\}_i$ such that 
\begin{equation*}
    \hat{s}(\omega_i) = s_i (\omega_i) \qquad \text{ for all } i \in \{1,2,...,n\},
\end{equation*}
and 
\[
\hat{s}_i(\omega) =
\begin{cases}
\bigcup_{j\neq i} s_j(\omega_i) & \text{if } \omega = \omega_i,\\[6pt]
\phi, & \omega \neq \omega_i.
\end{cases}
\]
Note that $\hat{s}_i \cap s_j$ reveals the true state for any $i,j$. Moreover, $\hat{s} \cap s_j (\omega_i)= \phi$ whenever $j \neq i$, implying that $\hat{s} \cap s_j$ also always reveals the true state for any $j$. That is, $\pi \lor \hat{\pi}$ is the fully revealing signal. \\

Is $\hat{\pi}$ the least informative (in the Blackwell sense) signal that, when combined with $\pi$, induces full information? Note that Blackwell informativeness is only a partial order, so a least element need not exist, and many signals are incomparable unless one refines the other. Rather than searching for a ``minimal'' way to induce full information for a given $\pi$, consider the weaker test: does $\hat{\pi}$ Blackwell-dominate every signal that, when combined with $\pi$, fully reveals the state? The following two examples show that the answer depends on the signal. First, consider the following signal 
\begin{center}
\scalebox{.9}{%
\begin{tabular}{c c c}
    $\omega_1$ & $\omega_2$ & $\omega_3$ \\
    \begin{tikzpicture}
        \node[left] at (0,0) {$\pi:$\quad};
        \draw[thick] (0,0) -- (4,0);
        \draw[blue]
            (0,0) -- (3,0) node[midway, yshift=-0.5cm] {$s_1$};
        \draw[brown]
            (3,0) -- (3.5,0) node[midway, yshift=-0.5cm] {$s_2$};
        \draw[red]
            (3.5,0) -- (4,0) node[midway, yshift=-0.5cm] {$s_3$};
        \filldraw[black] (0,0) circle (2pt);
        \filldraw[black] (3,0) circle (2pt);
        \filldraw[black] (3.5,0) circle (2pt);
        \filldraw[black] (4,0) circle (2pt);
    \end{tikzpicture} &
    \begin{tikzpicture}
        \draw[thick] (0,0) -- (4,0);
        \draw[brown]
            (0,0) -- (3,0) node[midway, yshift=-0.5cm] {$s_2$};
        \draw[blue]
            (3,0) -- (3.5,0) node[midway, yshift=-0.5cm] {$s_1$};
        \draw[red]
            (3.5,0) -- (4,0) node[midway, yshift=-0.5cm] {$s_3$};
        \filldraw[black] (0,0) circle (2pt);
        \filldraw[black] (3,0) circle (2pt);
        \filldraw[black] (3.5,0) circle (2pt);
        \filldraw[black] (4,0) circle (2pt);
    \end{tikzpicture}  & 
        \begin{tikzpicture}
        \draw[thick] (0,0) -- (4,0);
        \draw[red]
            (0,0) -- (3,0) node[midway, yshift=-0.5cm] {$s_3$};
        \draw[brown]
            (3,0) -- (3.5,0) node[midway, yshift=-0.5cm] {$s_2$};
        \draw[blue]
            (3.5,0) -- (4,0) node[midway, yshift=-0.5cm] {$s_1$};
        \filldraw[black] (0,0) circle (2pt);
        \filldraw[black] (3,0) circle (2pt);
        \filldraw[black] (3.5,0) circle (2pt);
        \filldraw[black] (4,0) circle (2pt);
    \end{tikzpicture}.\\
\end{tabular}
}
\end{center}
Its error-correcting refinement is given by 
\begin{center}
\scalebox{.9}{%
\begin{tabular}{c c c}
    $\omega_1$ & $\omega_2$ & $\omega_3$ \\
    \begin{tikzpicture}
        \node[left] at (0,0) {$\hat{\pi}:$\quad};
        \draw[thick] (0,0) -- (4,0);
        \draw[black]
            (0,0) -- (3,0) node[midway, yshift=-0.5cm] {$\hat{s}$};
        \draw[blue]
            (3,0) -- (3.5,0) node[midway, yshift=-0.5cm] {$\hat{s}_1$};
        \draw[blue]
            (3.5,0) -- (4,0) node[midway, yshift=-0.5cm] {$\hat{s}_1$};
        \filldraw[black] (0,0) circle (2pt);
        \filldraw[black] (3,0) circle (2pt);
        \filldraw[black] (3.5,0) circle (2pt);
        \filldraw[black] (4,0) circle (2pt);
    \end{tikzpicture} &
    \begin{tikzpicture}
        \draw[thick] (0,0) -- (4,0);
        \draw[black]
            (0,0) -- (3,0) node[midway, yshift=-0.5cm] {$\hat{s}$};
        \draw[brown]
            (3,0) -- (3.5,0) node[midway, yshift=-0.5cm] {$\hat{s}_2$};
        \draw[brown]
            (3.5,0) -- (4,0) node[midway, yshift=-0.5cm] {$\hat{s}_2$};
        \filldraw[black] (0,0) circle (2pt);
        \filldraw[black] (3,0) circle (2pt);
        \filldraw[black] (3.5,0) circle (2pt);
        \filldraw[black] (4,0) circle (2pt);
    \end{tikzpicture}  & 
        \begin{tikzpicture}
        \draw[thick] (0,0) -- (4,0);
        \draw[black]
            (0,0) -- (3,0) node[midway, yshift=-0.5cm] {$\hat{s}$};
        \draw[red]
            (3,0) -- (3.5,0) node[midway, yshift=-0.5cm] {$\hat{s}_3$};
        \draw[red]
            (3.5,0) -- (4,0) node[midway, yshift=-0.5cm] {$\hat{s}_3$};
        \filldraw[black] (0,0) circle (2pt);
        \filldraw[black] (3,0) circle (2pt);
        \filldraw[black] (3.5,0) circle (2pt);
        \filldraw[black] (4,0) circle (2pt);
    \end{tikzpicture}.  \\
\end{tabular}
}
\end{center}
The only way to obtain a signal that is less informative than $\hat{\pi}$ is by merging at least two messages of $\hat{\pi}$. Note that any signal obtained from $\hat{\pi}$ by merging messages does not fully reveal the information when combined with $\pi$. Therefore, $\hat{\pi}$ does not Blackwell dominate any signal that reveals full information when combined with $\pi$.\\

Second, consider the following signal 
\begin{center}
\scalebox{.9}{%
\begin{tabular}{c c c}
    $\omega_1$ & $\omega_2$ & $\omega_3$ \\
    \begin{tikzpicture}
        \node[left] at (0,0) {$\pi:$\quad};
        \draw[thick] (0,0) -- (4,0);
        \draw[blue]
            (0,0) -- (3,0) node[midway, yshift=-0.5cm] {$s_1$};
        \draw[brown]
            (3,0) -- (4,0) node[midway, yshift=-0.5cm] {$s_2$};
        \filldraw[black] (0,0) circle (2pt);
        \filldraw[black] (3,0) circle (2pt);
        \filldraw[black] (4,0) circle (2pt);
    \end{tikzpicture} &
    \begin{tikzpicture}
        \draw[thick] (0,0) -- (4,0);
        \draw[brown]
            (0,0) -- (3,0) node[midway, yshift=-0.5cm] {$s_2$};
        \draw[red]
            (3,0) -- (4,0) node[midway, yshift=-0.5cm] {$s_3$};
        \filldraw[black] (0,0) circle (2pt);
        \filldraw[black] (3,0) circle (2pt);
        \filldraw[black] (4,0) circle (2pt);
    \end{tikzpicture}  & 
        \begin{tikzpicture}
        \draw[thick] (0,0) -- (4,0);
        \draw[red]
            (0,0) -- (3,0) node[midway, yshift=-0.5cm] {$s_3$};
        \draw[blue]
            (3,0) -- (4,0) node[midway, yshift=-0.5cm] {$s_1$};
        \filldraw[black] (0,0) circle (2pt);
        \filldraw[black] (3,0) circle (2pt);
        \filldraw[black] (4,0) circle (2pt);
    \end{tikzpicture}.\\
\end{tabular}
}
\end{center}
Its error-correcting refinement is given by 
\begin{center}
\scalebox{.9}{%
\begin{tabular}{c c c}
    $\omega_1$ & $\omega_2$ & $\omega_3$ \\
    \begin{tikzpicture}
        \node[left] at (0,0) {$\hat{\pi}:$\quad};
        \draw[thick] (0,0) -- (4,0);
        \draw[black]
            (0,0) -- (3,0) node[midway, yshift=-0.5cm] {$\hat{s}$};
        \draw[blue]
            (3,0) -- (4,0) node[midway, yshift=-0.5cm] {$\hat{s}_1$};
        \filldraw[black] (0,0) circle (2pt);
        \filldraw[black] (3,0) circle (2pt);
        \filldraw[black] (4,0) circle (2pt);
    \end{tikzpicture} &
    \begin{tikzpicture}
        \draw[thick] (0,0) -- (4,0);
        \draw[black]
            (0,0) -- (3,0) node[midway, yshift=-0.5cm] {$\hat{s}$};
        \draw[brown]
            (3,0) -- (4,0) node[midway, yshift=-0.5cm] {$\hat{s}_2$};
        \filldraw[black] (0,0) circle (2pt);
        \filldraw[black] (3,0) circle (2pt);
        \filldraw[black] (4,0) circle (2pt);
    \end{tikzpicture}  & 
        \begin{tikzpicture}
        \draw[thick] (0,0) -- (4,0);
        \draw[black]
            (0,0) -- (3,0) node[midway, yshift=-0.5cm] {$\hat{s}$};
        \draw[red]
            (3,0) -- (4,0) node[midway, yshift=-0.5cm] {$\hat{s}_3$};
        \filldraw[black] (0,0) circle (2pt);
        \filldraw[black] (3,0) circle (2pt);
        \filldraw[black] (4,0) circle (2pt);
    \end{tikzpicture}.\\
\end{tabular}
}
\end{center}
Consider the signal $\tilde{\pi}$ obtained from $\hat{\pi}$ by merging  $\hat{s}_1,\hat{s}_2$, and $\hat{s}_3$ as one message $\tilde{s}$. That is  $\tilde{\pi} = \{\tilde{s},\hat{s}\}$.
\begin{center}
\scalebox{.9}{%
\begin{tabular}{c c c}
    $\omega_1$ & $\omega_2$ & $\omega_3$ \\
    \begin{tikzpicture}
        \node[left] at (0,0) {$\tilde{\pi}:$\quad};
        \draw[thick] (0,0) -- (4,0);
        \draw[black]
            (0,0) -- (3,0) node[midway, yshift=-0.5cm] {$\hat{s}$};
        \draw[cyan]
            (3,0) -- (4,0) node[midway, yshift=-0.5cm] {$\tilde{s}$};
        \filldraw[black] (0,0) circle (2pt);
        \filldraw[black] (3,0) circle (2pt);
        \filldraw[black] (4,0) circle (2pt);
    \end{tikzpicture} &
    \begin{tikzpicture}
        \draw[thick] (0,0) -- (4,0);
        \draw[black]
            (0,0) -- (3,0) node[midway, yshift=-0.5cm] {$\hat{s}$};
        \draw[cyan]
            (3,0) -- (4,0) node[midway, yshift=-0.5cm] {$\tilde{s}$};
        \filldraw[black] (0,0) circle (2pt);
        \filldraw[black] (3,0) circle (2pt);
        \filldraw[black] (4,0) circle (2pt);
    \end{tikzpicture}  & 
        \begin{tikzpicture}
        \draw[thick] (0,0) -- (4,0);
        \draw[black]
            (0,0) -- (3,0) node[midway, yshift=-0.5cm] {$\hat{s}$};
        \draw[cyan]
            (3,0) -- (4,0) node[midway, yshift=-0.5cm] {$\tilde{s}$};
        \filldraw[black] (0,0) circle (2pt);
        \filldraw[black] (3,0) circle (2pt);
        \filldraw[black] (4,0) circle (2pt);
    \end{tikzpicture}.\\
\end{tabular}
}
\end{center}
In this case, $\tilde{\pi}$  reveals full information when combined with $\pi$ and is Blackwell dominated by $\hat{\pi}$. Nevertheless, starting from $\hat{\pi}$ one can construct another signal that does not Blackwell dominate any signal that, when combined with $\pi$, fully reveals the state. Since signals are finite partitions, such a construction can be obtained by pairwise merging messages of $\hat{\pi}$. The following lemma states an important property of the error-correcting refinement. 
\begin{lemma}\label{generallemma}
Let $\Omega = \{\omega_1,\omega_2,\ldots,\omega_n\}$, $A = \{a_1,a_2,\ldots,a_n\}$, and  $u(a,\omega)$ be as defined in (\ref{diagonalmatrix}).
 Moreover, let $\pi \in \Pi$ and $\hat{\pi}$ be the error-correcting refinement of $\pi$.  Then, 
\begin{equation*}
    V(\hat{\pi} \mid \pi) \, \geq \,  V(\hat{\pi} \mid \pi') \quad \forall \pi' \in \Pi. 
\end{equation*}
\end{lemma}
The formal proof is in Appendix~\ref{appendix:generallemma}. 
Below, I will provide the intuition through an example where the decision maker has two types.
\subsubsection{Intuition of Lemma \ref{generallemma}}
Here I provide the intuition for a case where one type is fully uninformed and the other has some information; the same intuition extends to any two types. Consider the following two types.
\begin{itemize}
    \item[1.] \emph{Uninformed (U)}: takes the default action $a_j \in A$. That is,
    \begin{equation*}
        \mu(\omega_j) u(\omega_j) \geq \mu(\omega)\, u(\omega) \quad \forall \omega \in \Omega.
    \end{equation*}
    \item[2.] \emph{Partly informed (I)}: owns a signal $\pi \in \Pi$. 
\end{itemize}
Let $\hat{\pi}$ be the {\it error-correcting refinement} of $\pi$. In state $\omega_i$, the decision maker benefits whenever she is corrected from choosing the incorrect action. This occurs exactly when she receives signal $\hat{s}_i$, which happens with probability $\mu(\hat{s}_i \mid \omega_i)$. The marginal benefit from being corrected in state $\omega_i$ is $u_i$. Hence, the informed type’s valuation for the fully revealing signal is
\begin{equation*}
    V_I = \sum_{i=1}^n \mu(\omega_i) \, \mu(\hat{s}_i \mid \omega_i) \, u_i. 
\end{equation*}
Now let us analyze the uninformed type's benefit from mimicking the informed type. 
Each signal $\hat{s}_i$ fully reveals state $\omega_i$. Consider two cases.

\medskip
\noindent\textit{Case 1.} $\hat{s}$ induces action $a_j \in A$.  
In this case, the uninformed type's benefit cannot exceed $V_I$. If the state is $\omega_j$, she will not change her action and thus receives no benefit. In any other state $\omega \neq \omega_j$, she will change her action whenever the informed type does.

\medskip
\noindent\textit{Case 2.} $\hat{s}$ induces action $a_k \neq a_j$.  In state $\omega_k$, the uninformed type always takes the correct action, whereas in any any other state $\omega_\ell$, she takes the correct action only when she receives message $\hat{s}_\ell$. 
The uninformed type's benefit is given by
\[
V_U 
= \sum_{\ell \neq k} \mu(\omega_\ell)\, \mu(\hat{s}_\ell \mid \omega_\ell)\, u_\ell
    + \mu(\omega_k)\, u_k 
    - \mu(\omega_j)\, u_j \, \leq  \, \sum_{\ell \neq k} \mu(\omega_\ell)\, \mu(\hat{s}_\ell \mid \omega_\ell)\, u_\ell \, \leq \, V_I.
\]
The last inequality follows from the fact that 
$\mu(\omega_j)\,u_j \geq \mu(\omega)\,u(\omega)$ for all $\omega \in \Omega$. \\
\\
The following proposition shows that, with diagonal payoffs, the data broker extracts the efficient surplus. 

\begin{proposition}\label{generalproposition}
Assume $\Omega = \{\omega_1, \omega_2, \dots, \omega_n\}$ and $A = \{a_1, a_2, \dots, a_n\}$, $\mathcal{T} = \Pi$ and $F \sim \mathcal{T}$. The data broker extracts the efficient surplus from the decision maker. In particular, the optimal menu is given by  
\begin{equation*}
    \mathcal{M} = \{(\hat{\pi}, t(\pi)) \mid \pi \in \Pi\},
\end{equation*}
where 
\begin{equation*}
    t(\pi) = \overline{U} - U(\pi).  
\end{equation*}
\end{proposition}
The proof is in appendix~\ref{appendix:generalproposition}. Under diagonal payoffs (state–action matching), any mismatch between the chosen action and the realized state yields the same baseline payoff. As a result, the marginal value of being corrected depends only on the state and is independent of the action the decision maker would otherwise take.
By contrast, under general payoff structures, the marginal value of being corrected typically depends on the default action the decision maker is taking. The example below shows that, in such environments, efficient surplus extraction cannot be guaranteed.

\subsubsection{Example}
The following example illustrates the case in which the data broker is not able to extract the efficient surplus from the decision maker. Consider the following payoff function represented by upper triangular matrix:
\[
\begin{array}{c|ccc}
   & \omega_1 & \omega_2 & \omega_3 \\ \hline
a_1 & u_{1} & u_1 & u_1\\
a_2 & 0 & u_{2} & u_2\\
a_3 & 0 & 0 & u_3\\
\end{array}
\]
Assume that $u_1 < u_2 < u_3$ and that the prior $\mu= (\mu_1,\mu_2,\mu_3)$ is such that the decision maker is indifferent between taking any action $a \in A$. That is 
    \begin{equation*}
        u_1 = (\mu_2  + \mu_3 )\, u_2 = \mu_3 \, u_3.
    \end{equation*}
Moreover, assume $\mu_i > 0$ for $i=1,2,3$. 
Consider two types for the decision maker i) fully uninformed, and ii) partially informed that owns a signal that only fully reveals when the state is $\omega_3$. 

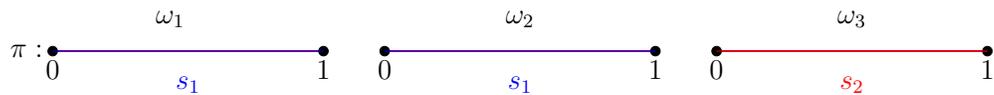
\begin{center}
\scalebox{.9}{%
\begin{tabular}{c c c}
    $\omega_1$ & $\omega_2$ & $\omega_3$ \\
    \begin{tikzpicture}
        \node[left] at (0,0) {$\pi:$\quad};
        \draw[thick] (0,0) -- (4,0);
        \draw[purple, thick] (0,0) -- (4,0);
        \draw[blue]
            (0,0) -- (4,0) node[midway, yshift=-0.5cm] {$s_1$};
        \filldraw[black] (0,0) circle (2pt) node[below] {$0$};
        \filldraw[black] (4,0) circle (2pt) node[below] {$1$};
    \end{tikzpicture} &
    \begin{tikzpicture}
        \draw[thick] (0,0) -- (4,0);
        \draw[purple, thick] (0,0) -- (4,0);
        \draw[blue]
            (0,0) -- (4,0) node[midway, yshift=-0.5cm] {$s_1$};
        \filldraw[black] (0,0) circle (2pt) node[below] {$0$};
        \filldraw[black] (4,0) circle (2pt) node[below] {$1$};
    \end{tikzpicture} & 
        \begin{tikzpicture}
        \draw[thick] (0,0) -- (4,0);
        \draw[purple, thick] (0,0) -- (4,0);
        \draw[red]
            (0,0) -- (4,0) node[midway, yshift=-0.5cm] {$s_2$};
        \filldraw[black] (0,0) circle (2pt) node[below] {$0$};
        \filldraw[black] (4,0) circle (2pt) node[below] {$1$};
    \end{tikzpicture} \\
\end{tabular}
}
\end{center}
A $\pi_{\phi}-$minimal complementary signal of $\pi$ is given by 
\begin{center}
\scalebox{.9}{%
\begin{tabular}{c c c}
    $\omega_1$ & $\omega_2$ & $\omega_3$ \\
    \begin{tikzpicture}
        \node[left] at (0,0) {$\pi^c:$\quad};
        \draw[thick] (0,0) -- (4,0);
        \draw[purple, thick] (0,0) -- (4,0);
        \draw[brown]
            (0,0) -- (4,0) node[midway, yshift=-0.5cm] {$x$};
        \filldraw[black] (0,0) circle (2pt) node[below] {$0$};
        \filldraw[black] (4,0) circle (2pt) node[below] {$1$};
    \end{tikzpicture} &
    \begin{tikzpicture}
        \draw[thick] (0,0) -- (4,0);
        \draw[purple, thick] (0,0) -- (4,0);
        \draw[green]
            (0,0) -- (4,0) node[midway, yshift=-0.5cm] {$y$};
        \filldraw[black] (0,0) circle (2pt) node[below] {$0$};
        \filldraw[black] (4,0) circle (2pt) node[below] {$1$};
    \end{tikzpicture} & 
    \begin{tikzpicture}
        \draw[thick] (0,0) -- (4,0);
        \draw[blue, thick] (0,0) -- (1,0);
        \draw[red, thick] (1,0) -- (4,0);
        \draw[brown]
            (0,0) -- (1,0) node[midway, yshift=-0.5cm] {$x$};
        \draw[green]
            (1,0) -- (4,0) node[midway, yshift=-0.5cm] {$y$};
        \filldraw[black] (1,0) circle (2pt);
        \filldraw[black] (0,0) circle (2pt) node[below] {$0$};
        \filldraw[black] (4,0) circle (2pt) node[below] {$1$};
        \filldraw[black] (1,0) circle (2pt) node[below] {$a$};
    \end{tikzpicture} \\
\end{tabular}
}
\end{center}
Since $\pi^c$ must refine $s_1$, it must separate $\omega_1$ from $\omega_2$. Because $\omega_3$ is fully revealed under $\pi$, we can, without loss of generality, let $\pi^c$ split $\omega_3$ by sending message $x$ on $[0,a]$ and $y$ on $[a,1]$. We show that both $x$ and $y$ are used in state $\omega_3$, i.e., $a\in(0,1)$. \\
\\
The cutoff $a$ solves
\begin{equation*}
    \min_{a\in[0,1]}\; \mu(x)\,u(x)\;+\;\mu(y)\,u(y). 
\end{equation*}
That is, 
\begin{equation*}
    \min_{a\in[0,1]}\Big[
\max\{(\mu_1+a\mu_3)u_1,\; a\mu_3 u_3\}
+
\max\{(\mu_2+(1-a)\mu_3)u_2,\; (1-a)\mu_3 u_3\}
\Big].
\end{equation*}
Since $(\mu_2+\mu_3)u_2=\mu_3u_3$, we have $(\mu_2+(1-a)\mu_3)u_2>(1-a)\mu_3u_3$ for any $a>0$. It is not optimal to have $(\mu_1+a\mu_3)u_1<a\mu_3u_3$. Otherwise, decreasing $a$ by some $\varepsilon>0$ decreases the first term by $\varepsilon\mu_3u_3$ and increases the second term by $\varepsilon\mu_3u_2$. Hence, $(\mu_1+a\mu_3)u_1\ge a\mu_3u_3$. 

Optimality requires $(\mu_1+a\mu_3)u_1=a\mu_3u_3$. To see this, assume $(\mu_1+a\mu_3)u_1>a\mu_3u_3$. Then $a<1$. Increasing $a$ by some $\varepsilon>0$ increases the first term by $\varepsilon\mu_3u_1$ and decreases the second term by $\varepsilon\mu_3u_2$. Thus, one optimal solution sets $a$ so that $(\mu_1+a\mu_3)u_1=a\mu_3u_3$, i.e.,
\[
a=\frac{\mu_1u_1}{\mu_3(u_3-u_1)}\in(0,1),
\]
since $u_1(\mu_1+\mu_3)<\mu_3u_3$.\\

The DM is indifferent between taking actions $a_1$ or $a_3$ after receiving signal realization $x$. Low type's benefit from $\pi^c$ is.
\begin{equation*}
    V(\pi^c \mid \pi)  =  \mu(\omega_2) \, (u_2-u_1).
\end{equation*}
High type's benefit from $\pi^c$ is
\begin{equation*}
\begin{split}
    V(\pi^c \mid \pi_\phi)  = \mu(\omega_2) \, (u_2-u_1) + \mu(\omega_3,y) \, (u_2 -u_1) >    V(\hat{\pi} \mid \pi).
\end{split}
\end{equation*}
since $\mu(\omega_3,y) > 0$. Proposition~\ref{fullsurplusextraction} implies that the data broker cannot extract the efficient surplus extraction in this environment.\\

Observe that the partially informed type already knows state $\omega_3$ and therefore ignores the messages that $\pi^c$ sends in that state. Nevertheless, although she ignores message $y$ in $\omega_3$, the uninformed type still benefits from its availability. Moreover, when following $\pi^c$, the uninformed decision maker never selects the optimal action in state $\omega_3$. Let $a^*(\omega_3)\in\arg\max_{a\in A} u(\omega_3,a)$ denote the optimal action, let $a_0\in A\setminus\{a^*(\omega_3)\}$ be the action she would take under $\pi_{\phi}$ alone in $\omega_3$, and let $a_1\in A\setminus\{a^*(\omega_3)\}$ be the action she takes when obeying $\pi^c$ in $\omega_3$. By the payoff assumptions, $u(\omega_3,a_1) > u(\omega_3,a_0)$ and $a_1 \neq a^*(\omega_3)$. Thus, in $\omega_3$ she attains a strictly higher payoff despite moving between two suboptimal actions.

\section{Conclusion}
This paper studies how a data broker sells information, in the form of a signal, to a decision maker who privately owns another signal. Unlike prior work on selling supplemental information, the decision maker’s private type is itself a signal structure—that is, a distribution over her interim beliefs. The analysis focuses on environments in which the decision maker purchases information before her private signal is realized, implying that the data broker must screen over all feasible signal structures. I allow the signals the data broker can sell to contain messages that are arbitrarily correlated with the ones the decision maker privately owns. Considering such a rich signal space plays a central role in the design of the optimal menu. Moreover, modeling information as a signal aligns with practical settings in which a supplemental dataset is appended to an existing one.\\

The data broker extract the efficient surplus in a large class of environments.  Comparing this setting with the case in which the data broker can sell only signals that are not correlated highlights an important distinction. The decision maker can end up fully informed even after purchasing partial information, if the signals the data broker can provide are correlated with the signal the decision maker privately owns. Moreover, the signals sold take the form of action recommendations, so the data broker can condition information on the action the decision maker intends to take.
Regarding efficiency, the ability of the data broker to sell correlated signals increases overall efficiency but reduces the decision maker’s information rents.  \\


\newpage
\bibliographystyle{chicago}
\bibliography{bib}

\clearpage

\section{Appendix} \label{sec:appendixa}
\addcontentsline{toc}{section}{appendix A}

\subsection{Proof of Lemma~\ref{reduceproblem}}\label{appendix:reduceproblem}

\begin{proof}
Assume $\{(\sigma_H,t_H), (\sigma_L,t_L)\}$ is an optimal allocation. 
\begin{itemize}
    \item [a.] Assume by contradiction that $U(\sigma_H \lor \pi_H) < \overline{U}$. So, there exists $\sigma_H' \in \Pi$ such that $V_H(\sigma_H') > V_H(\sigma_H)$. 
    
    If $IR_H$ binds, consider changing the high type's allocation to $(\sigma_H', t_H')$, where  $t_H' = V_H(\sigma_H')$. High type's utility remains unchanged. So, $IR_H$ and $IC_H$ hold. Moreover, $IR_L$ remains unchanged. The high type has a higher willingness to pay for the full information than the low type, hence we have 
    \begin{equation*}
        V_L(\sigma_H') - t_H' \leq V_H(\sigma_H')  - t_H' = 0,
    \end{equation*}
    implying that $IC_L$ is not violated.

    If $IR_H$ does not bind, then $IC_H$ must bind. Otherwise, the data broker can attain a higher profit by raising $t_H$ by a small amount. Consider now changing the high type's allocation to $(\sigma_H', t_H')$ so that $\sigma_H' \lor \pi_H = \overline{\pi}$ and
    \begin{equation*}
        t_H' = V_H(\sigma_H') - \big( V_H(\sigma_L)-t_L \big). 
    \end{equation*}
    $IR_L$ does not change. Clearly, $IC_H$ is not violated since $ V_H(\sigma_H') - t_H' = V_H(\sigma_L) - t_L.$ Moreover, $IR_H$ is not violated because 
    \begin{equation*}
        V_H(\sigma_H') - t_H' = V_H(\sigma_L) - t_L =  V_H(\sigma_H) - t_H \geq 0. 
    \end{equation*}
    The second equality comes from the fact that given $\{(\sigma_H,t_H), (\sigma_L,t_L)\}$,  $IC_H$ binds. If $IC_L$ holds, then $\{(\sigma_H',t_H'), (\sigma_L,t_L)\}$ is feasible and yields higher profits. \footnote{$t_H' = V_H(\sigma_H') - (V_H(\sigma_L) - t_L) > V_H(\sigma_H) - (V_H(\sigma_L) - t_L) = V_H(\sigma_H)  - (V_H(\sigma_H)-t_H)$. The last equality comes from the fact that $IC_H$ binds. That is, $t_H' > t_H. $} If $IC_L$ does not hold, then we have 
    \begin{equation*}
          V_L(\sigma_L) - t_L <  V_L(\sigma_H') - t_H'. 
    \end{equation*}
    Note that $t_H' \geq  t_L$ since $V_H(\sigma_H') \geq V_H(\sigma_L)$. But then, selling the same allocation $(\sigma_H', t_H')$ to both types yields higher profits. 

\item [b.] By contradiction, assume that $IR_L$ does not bind. Then $IC_L$ and $IR_H$ bind. But then,
\begin{equation*}
\begin{aligned}
V_L(\sigma_L) - t_L
&= V_L(\sigma_H) - t_H \quad \text{since $IC_L$ binds,} \\
&= V_L(\sigma_H) - V_H(\sigma_H) \quad \text{since $IR_H$ binds,} \\
&= V_L(\sigma_H) - V_H(\overline{\pi}) \quad \text{from part (a),} \\
&\le 0.
\end{aligned}
\end{equation*}
Thus, $IR_L$ binds.

\item [c.] By contradiction, assume that $IC_H$ does not bind. Then $IR_H$ must bind (if neither binds, this cannot be optimal, since the data broker could raise the price $t_H$). Since $IR_H$ binds, we have $V_H(\sigma_H) - t_H = 0$. Hence, the $IC_H$ constraint writes
\begin{equation*} 
V_H(\sigma_L) - t_L < V_H(\sigma_H)-t_H = 0. 
\end{equation*}
From part (b), we have $t_L = V_L(\sigma_L)$. Hence, $V_H(\sigma_L) < V_L(\sigma_L)$. Since the data broker cannot extract the efficient surplus, we have $\sigma_L \lor \pi_L \neq \overline{\pi}$. Otherwise, since $IR_L$ binds, $IR_H$ would be slack. Thus, $V_L(\sigma_L) < V_L(\overline{\pi})$. But then, since $IC_H$ is slack, consider giving the low type a better signal $\sigma_L'$ that will improve his utility and increase $t_L$ without violating the $IC_H$ constraint. That is, we can choose $\sigma_L'$ so that $V_H(\sigma_L)$ will increase by a very small amount. \\

More formally, let us first show that for any $\varepsilon > 0$, there exists $\sigma_L' \in \Pi$ such that $V_H(\sigma_L') \leq V_H(\sigma_L) + \varepsilon$ and $V_L(\sigma_L') > V_L(\sigma_L)$. Since $U(\sigma_L \lor \pi_L) < \overline{U}$, there exist $s\in \sigma_L$ and $s' \in \pi_L$ such that $s \cap s'$ induces a suboptimal action in some state $\omega$. Then, obtain $\sigma_L'$ from $\sigma_L$ by refining $s$ in a way that reveals the true state $\omega$ only on a very small subset in $(s' \cap s)(\omega)$. This region can be chosen arbitrarily small so that $V_H(\sigma_L') \leq V_H(\sigma_L) + \varepsilon$. \\

Then, consider $\sigma_L'$ obtained from $\sigma_L$ as above, and let $t_L' = t_L + V_L(\sigma_L') - V_L(\sigma_L)$. Note that $IR_L$, $IR_H$, and $IC_L$ will not be violated. Moreover, we can choose $\varepsilon >0$ such that it will not violate $IC_H$.

\end{itemize}
\end{proof}

\subsection{Proof of Proposition~\ref{fullsurplusextraction}} \label{appendix:fullsurplusextraction}

\begin{proof}
Assume that
\[
V_L(\pi^c_L(\pi_H)) \geq V_H(\pi^c_L(\pi_H)).
\]
Then the data broker can extract the efficient surplus as follows. The data broker sells the fully revealing signal to the high type at a price equal to his willingness to pay. The low type does not mimic the high type. The data broker sells the signal $\pi^c_L(\pi_H)$ to the low type at a price equal to his willingness to pay. By the assumption above, the high type has no incentive to mimic the low type. Therefore, the resulting menu is incentive compatible and allows the data broker to extract the efficient surplus.\\

To prove the converse, assume by contradiction that $V_L(\pi^c_L(\pi_H)) < V_H(\pi^c_L(\pi_H))$ and the data broker extracts the efficient surplus. From (\ref{quasiminimalcomplementarysignal}), it follows that for any $\sigma_L$ with $U(\sigma_L \lor \pi_L) = \overline{U}$, we have $V_L(\sigma_L) = V_L(\pi^c_L(\pi_H))$ and $U(\pi_H \lor \pi^c_L(\pi_H) ) \leq U(\pi_H \lor \sigma_L)$. That is,
\begin{equation*}
    V_H(\pi^c_L(\pi_H)) \leq V_H(\sigma_L). 
\end{equation*}
From the supposition that $V_L(\pi^c_L(\pi_H)) < V_H(\pi^c_L(\pi_H))$, we have 
\begin{equation*}
   V_L(\sigma_L)= V_L(\pi^c_L(\pi_H)) < V_H(\pi^c_L(\pi_H)) \leq V_H(\sigma_L).
\end{equation*}
So, given any signal $\sigma_L$ with $U(\sigma_L \lor \pi_L) = \overline{U}$, we have $ V_L(\sigma_L) < V_H(\sigma_L)$. However, this contradicts efficient surplus extraction, because the high type can obtain strictly positive rents by mimicking the low type rather than choosing his designated signal.
\end{proof}

\subsection{Proof of Lemma~\ref{binary}} \label{appendix:binary}
\begin{proof}
Let $\pi = \{s_1, s_2, \dots, s_k\}$ for some $k \in \mathbb{N}$.
Define $\tilde{s}_1$ as the collection of all messages that induce action $a_1$, and $\tilde{s}_2$ as the collection of all messages that induce action $a_2$. Since $\pi$ does not fully reveal the true state, assume without loss of generality that $\mu(\tilde{s}_2 \mid \omega_1) > 0$. That is, the decision maker takes the wrong action in state $\omega_1$ with some positive probability.\\

First, the minimal complementary signal must contain at least two messages to identify the true state when the signal realization is  $\tilde{s}_2$. That is, there exist messages $\hat{s}_1$ and $\hat{s}_2$ such that 
\begin{equation*}
    \tilde{s}_2(\omega_1) \subseteq \hat{s}_2(\omega_1) \quad  \text{ and } \quad \tilde{s}_2(\omega_2) \subseteq \hat{s}_1(\omega_2) .
\end{equation*}
This implies that $\mu(\omega_1 \mid \tilde{s}_2 \cap \hat{s}_2) = 1$ and $\mu(\omega_2 \mid \tilde{s}_2 \cap \hat{s}_1) = 1$. In other words, $\hat{s}_1 \cap \tilde{s}_2$ and $\hat{s}_2 \cap \tilde{s}_2$ fully reveal the true state. Hence, the minimal complementary signal must have at least two messages. \\

Second, it is without loss to restrict attention to minimal complementary signals with at most two messages. If $\mu(\tilde{s}_1 \mid \omega_2) >0$, then to identify the true state conditional on $\tilde{s}_1$, two messages $s_1'$ and $s_2'$ with $\tilde{s}_1(\omega_1) \subseteq s_1'(\omega_1)$ and $\tilde{s}_1(\omega_2) \subseteq s_2'(\omega_2)$ are needed.   Note that it is weakly less informative to merge $s_1'$ with $\hat{s}_1$ as one signal, and  $s_2'$ with $\hat{s}_2$ as another signal.

\end{proof}

\subsection{Proof of Lemma~\ref{mainlemma}}\label{appendix:mainlemma}

\begin{proof}
Take $\pi,\pi' \in \Pi$. Let $s_1$ be the union of all messages in $\pi$ that induce action $a_1$, and $s_2$ be the union of all messages in $\pi$ that induce action $a_2$. Moreover, let  $\pi^c= \{\hat{s}_1, \hat{s}_2\}$ be the minimal complementary signal of $\pi$ where
\begin{equation*}
    \hat{s}_1(\omega_1) = s_1(\omega_1), \quad \hat{s}_1(\omega_2) = s_2(\omega_2) \quad \text{ and } \quad  \hat{s}_2(\omega_1) = s_2(\omega_1), \quad \hat{s}_2(\omega_2) = s_1(\omega_2). 
\end{equation*}
That is, receiving $\hat{s}_1$ can be interpreted as receiving message "don't switch", whereas receiving $\hat{s}_2$ can be interpreted as receiving message "switch". The benefit of type $\pi$ from $\pi^c$ is 
\begin{equation*}
\begin{split}
    V(\pi^c \mid \pi) & = U(\pi^c \lor \pi) - U(\pi)  \\
    & =\sum_{\omega \in \Omega_1} \mu(\omega)\, u_1(\omega) + \sum_{\omega \in \Omega_2} \mu(\omega)\, u_2(\omega)  - \sum_{\omega \in \Omega} \mu(\omega) \, [\mu(s_1 \mid \omega) \, u_1(\omega) + \mu(s_2 \mid \omega)\, u_2(\omega)]\\
    & = \sum_{\omega \in \Omega_1} \mu(\omega)\, \mu(s_2 \mid \omega) \, (u_1(\omega) - u_2(\omega)) + \sum_{\omega \in \Omega_2} \mu(\omega) \, \mu(s_1 \mid \omega) \, (u_2(\omega) - u_1(\omega)) \\
    & = \sum_{\omega \in \Omega_1}  \mu(s_2, \omega) \, (u_1(\omega) - u_2(\omega)) + \sum_{\omega \in \Omega_2} \mu(s_1, \omega) \, (u_2(\omega) - u_1(\omega)) \\
    & = \sum_{\omega \in \Omega_1} \mu(\hat{s}_2, \omega) \, (u_1(\omega) - u_2(\omega)) + \sum_{\omega \in \Omega_2} \mu(\hat{s}_2, \omega) \, (u_2(\omega) - u_1(\omega)). 
\end{split}
\end{equation*}
The last equality follows from the fact that $\hat{s}_2(\omega_1) = s_2(\omega_1)$ and $\hat{s}_2(\omega_2) =s_1(\omega_2)$. The benefit of type $\pi'$ from $\pi^c$ is
\begin{equation*}
\begin{split}
    V(\pi^c \mid \pi') & = U(\pi^c \lor \pi') - U(\pi')\\
    & = \sum_{s'\in \pi'} \sum_{\hat{s} \in \pi^c} \mu(s' \cap \hat{s}) \, u(s' \cap \hat{s}) - \sum_{s' \in \pi'} \mu(s') \, u(s')\\
    & = \sum_{s' \in \pi'} \mu(s') \bigg[ \sum_{\hat{s} \in \pi^c} \mu(\hat{s} \mid s') \, u(s' \cap \hat{s}) - u(s') \bigg] \\
    & = \sum_{s' \in \pi'} \mu(s') \, V(\pi^c \mid s'). 
\end{split}
\end{equation*}
Take any $s \in \pi'$.  If $\hat{s}_1 \cap s $ and $\hat{s}_2 \cap s$ induce the same action,  then $V(\pi^c \mid s)  = 0$.  Now, consider the following two cases:\\
\\
{\it Case 1.} Assume $\hat{s}_1 \cap s$ induces action $a_1$ and $\hat{s}_2 \cap s$ induces action $a_2$. We have 
\begin{equation*}
\begin{aligned}
V(\pi^c \mid s)
&= \mu(\hat{s}_1\mid s)\sum_{\omega\in\Omega}\mu(\omega\mid \hat{s}_1,s)\,u_1(\omega)
 + \mu(\hat{s}_2\mid s)\sum_{\omega\in\Omega}\mu(\omega\mid \hat{s}_2,s)\,u_2(\omega)
 - \max_{a_1,a_2}\sum_{\omega\in\Omega}\mu(\omega\mid s)\,u(a,\omega)\\
&\le \mu(\hat{s}_1\mid s)\sum_{\omega\in\Omega}\mu(\omega\mid \hat{s}_1,s)\,u_1(\omega)
 + \mu(\hat{s}_2\mid s)\sum_{\omega\in\Omega}\mu(\omega\mid \hat{s}_2,s)\,u_2(\omega)
 - \sum_{\omega\in\Omega}\mu(\omega\mid s)\,u_1(\omega)\\
&= \mu(\hat{s}_2\mid s)\sum_{\omega\in\Omega}\mu(\omega\mid \hat{s}_2,s)\big(u_2(\omega)-u_1(\omega)\big)\\
&= \mu(\hat{s}_2\mid s)\sum_{\omega\in\Omega_2}\mu(\omega\mid \hat{s}_2,s)\big(u_2(\omega)-u_1(\omega)\big)
 + \mu(\hat{s}_2\mid s)\sum_{\omega\in\Omega_1}\mu(\omega\mid \hat{s}_2,s)\big(u_2(\omega)-u_1(\omega)\big)\\
&\le \mu(\hat{s}_2\mid s)\sum_{\omega\in\Omega_2}\mu(\omega\mid \hat{s}_2,s)\big(u_2(\omega)-u_1(\omega)\big).
\end{aligned}
\end{equation*}
The last inequality follows because $u_1(\omega) > u_2(\omega)$ for all $\omega \in \Omega_1$.\\
\\
{\it Case 2.} Assume  $\hat{s}_1 \cap s$ induces action $a_2$ and $\hat{s}_2 \cap s$ induces action $a_1$. We have
\begin{equation*}
\begin{aligned}
V(\hat{\pi}\mid s)
&=  \mu(\pi^c \mid s)\sum_{\omega\in\Omega}\mu(\omega\mid \hat{s}_2,s)\,u_1(\omega) +\mu(\hat{s}_1\mid s)\sum_{\omega\in\Omega}\mu(\omega\mid \hat{s}_1,s)\,u_2(\omega)
 - \max_{a_1,a_2}\sum_{\omega\in\Omega}\mu(\omega\mid s)\,u(a,\omega)\\
&\le \mu(\hat{s}_2\mid s)\sum_{\omega\in\Omega}\mu(\omega\mid \hat{s}_2,s)\,u_1(\omega) +\mu(\hat{s}_1\mid s)\sum_{\omega\in\Omega}\mu(\omega\mid \hat{s}_1,s)\,u_2(\omega)
 - \sum_{\omega\in\Omega}\mu(\omega\mid s)\,u_2(\omega)\\
&= \mu(\hat{s}_2\mid s)\sum_{\omega\in\Omega}\mu(\omega\mid \hat{s}_2,s)\,u_1(\omega) - \mu(\hat{s}_2\mid s)\sum_{\omega\in\Omega}\mu(\omega\mid \hat{s}_2,s)\,u_2(\omega)\\
& =\mu(\hat{s}_2\mid s)\sum_{\omega\in\Omega}\mu(\omega\mid \hat{s}_2,s)\,(u_1(\omega) - u_2(\omega))\\
& =\mu(\hat{s}_2\mid s)\sum_{\omega\in\Omega_1}\mu(\omega\mid \hat{s}_2,s)\,(u_1(\omega) - u_2(\omega)) + \mu(\hat{s}_2\mid s)\sum_{\omega\in\Omega_2}\mu(\omega\mid \hat{s}_2,s)\,(u_1(\omega) - u_2(\omega))\\
& \leq \mu(\hat{s}_2\mid s)\sum_{\omega\in\Omega_1}\mu(\omega\mid \hat{s}_2,s)\,(u_1(\omega) - u_2(\omega)).
\end{aligned}
\end{equation*}
The last inequality follows because $u_1(\omega) < u_2(\omega)$ for all $\omega \in \Omega_2$. So, overall we have 
\begin{equation*}
\begin{split}
    V(\pi^c \mid s) & \leq \mu(\hat{s}_2\mid s)\sum_{\omega\in\Omega_2}\mu(\omega\mid \hat{s}_2,s)\big(u_2(\omega)-u_1(\omega)\big) + \mu(\hat{s}_2\mid s)\sum_{\omega\in\Omega_1}\mu(\omega\mid \hat{s}_2,s)\,(u_1(\omega) - u_2(\omega))\\
    & = \mu(\hat{s}_2\mid s) \bigg(\sum_{\omega\in\Omega_2}\mu(\omega\mid \hat{s}_2,s)\big(u_2(\omega)-u_1(\omega)\big) + \sum_{\omega \in \Omega_1}\mu(\omega\mid \hat{s}_2,s)\,(u_1(\omega) - u_2(\omega))\bigg).\\
\end{split}
\end{equation*}
Aggregating over all $s' \in \pi'$ yields the following inequality.
\begin{equation*}
\begin{split}
V(\pi^c \mid \pi') & = \sum_{s \in \pi'} \mu(s) V(\pi^c \mid s) \\
& \le \sum_{s \in \pi' } \mu(s) \mu(\hat{s}_2\mid s)\bigg[\sum_{\omega\in\Omega_2}\mu(\omega\mid \hat{s}_2,s)\big(u_2(\omega)-u_1(\omega)\big) + \sum_{\omega\in\Omega_1}\mu(\omega\mid \hat{s}_2,s)\,(u_1(\omega) - u_2(\omega))\bigg].
\end{split}
\end{equation*}
Note that since
\begin{equation*}
\mu(s) \mu(\hat{s}_2 \mid s) \mu(\omega \mid \hat{s}_2, s) = \mu(\omega, \hat{s}_2, s) .
\end{equation*}
we have 
\begin{equation*}
\begin{split}
V(\pi^c \mid \pi') 
& \le \sum_{s \in \pi' } \bigg[\sum_{\omega\in\Omega_2}\mu(\omega, \hat{s}_2, s)  \big(u_2(\omega)-u_1(\omega)\big) + \sum_{\omega\in\Omega_1}\mu(\omega, \hat{s}_2, s) \,(u_1(\omega) - u_2(\omega))\bigg]  \\
& = \sum_{\omega\in\Omega_2}  \sum_{s \in \pi' } \mu(\omega, \hat{s}_2, s)  \big(u_2(\omega)-u_1(\omega)\big) + \sum_{\omega\in\Omega_1}  \sum_{s \in \pi' }\mu(\omega, \hat{s}_2, s) \,(u_1(\omega) - u_2(\omega)) . \\
\end{split}
\end{equation*}
Since $\pi'$ partitions $\hat{s}_2(\omega)$, we have 
\begin{equation*}
    \sum_{s \in \pi'} \mu(\omega, \hat{s}_2, s) = \mu(\omega, \hat{s}_2). 
\end{equation*}
That is,
\begin{equation*}
\begin{split}
V(\pi^c \mid \pi') 
& \le   \sum_{\omega\in\Omega_2} \mu(\omega, \hat{s}_2)  \big(u_2(\omega)-u_1(\omega)\big) + \sum_{\omega\in\Omega_1}  \mu(\omega, \hat{s}_2) \,(u_1(\omega) - u_2(\omega))  \\
& = V(\pi^c \mid \pi).
\end{split}
\end{equation*}

\end{proof}

\subsection{Proof of Proposition~\ref{mainproposition}}\label{appendix:mainproposition}
\begin{proof}
Consider the following menu that consists of the collections of all minimal complementary signals 
\begin{equation*}
    \mathcal{M} = \{(\pi^c, t({\pi}))\}_{\pi \in \Pi}, \quad \text{ where }  \quad t(\pi) = U(\pi \lor \pi^c) - U(\pi).
\end{equation*}
Note that this is the maximum possible revenue the data broker can achieve, since each type ends up with full information and pays their maximum valuation that information. Lemma \ref{mainlemma} shows that this menu is incentive compatible, as no type benefits from purchasing another type’s minimal complementary signal.

\end{proof}

\subsection{Proof of Lemma ~\ref{generallemma}} \label{appendix:generallemma}
\begin{proof}
Let $\pi \in \Pi$. Assume $\{s_1,s_2,...,s_n\}$ are such that such that $s_1$ is the union of all messages in $\pi$ that induce action $a_1$, and for any $i \in \{2,...,n\}$, $s_i$ is the union of all messages in $\pi$ that induce action $a_{i}$, and do not induce action $a_{k}$ for any $k < i$. That is,  $u(s_i) = \mu(\omega_i \mid s_i) \, u_i$ for each $i \in \{1,2,...,n\}$. Let $\hat{\pi} = \{\hat{s}_1,...,\hat{s}_n, \hat{s}\}$ be the error-correcting refinement of $\pi$.  First, note that
First, note that
\begin{equation}
\begin{split}
    V(\hat{\pi} \mid \pi) & =  \sum_{i=1}^n \mu(\omega_i) \, u_i - \sum_{i=1}^n \mu(s_i) \, \mu(\omega_i \mid s_i) \, u_i \\
    & = \sum_{i=1}^n \mu(\omega_i) \bigg(1-\mu(s_i \mid \omega_i) \bigg) \, u_i \\
    & = \sum_{i=1}^n \mu(\omega_i) \, \bigg(\sum_{i \neq j} \mu(s_j \mid \omega_i) \bigg) u_i \\
    & = \sum_{i=1}^n \mu(\omega_i) \, \mu(\hat{s}_i \mid \omega_i) \, u_i. 
\end{split}
\end{equation}
Second, note that 
\begin{equation} \label{eqgeneral}
\begin{split}
  V(\hat{\pi} \mid \pi') & = \sum_{s' \in \pi'} \sum_{s'' \in \hat{\pi}} \mu(s' \cap s'') \, u(s' \cap s'') - \sum_{s' \in \pi'} \mu(s') \, u(s') \\
  & =  \sum_{s' \in \pi'} \mu(s') \bigg[\sum_{s'' \in \hat{\pi}} \mu(s'' \mid s') \, u(s' \cap s'') - u(s') \bigg] \\
  & =  \sum_{s' \in \pi'} \mu(s') V(\hat{\pi} \mid s'). 
\end{split}
\end{equation}
Fix any $s' \in \pi'$. Let $a_\ell$ be the optimal action upon receiving message $s'$. That is 
\begin{equation*}
    \mu(\omega_\ell \mid s') \, u_\ell \geq \mu(\omega_i \mid s') \, u_i, \quad \forall i. 
\end{equation*}
Consider two cases 
\begin{itemize}
    \item [i)] $\hat{s} \cap s'$ induces action $a_\ell$.  
    \begin{equation*}
    \begin{split}
        V(\hat{\pi} \mid s') & = \sum_{s'' \in \hat{\pi}} \mu(s'' \mid s') \, u(s' \cap s'') - u(s') \\
        & = \sum_{i=1}^n \mu(\hat{s}_i \mid s') u(s' \cap \hat{s}_i)   + \mu(\hat{s} \mid s') \, \mu(\omega_\ell \mid s'\cap \hat{s})\, u_\ell  - \mu(\omega_\ell \mid s')\, u_\ell \\
        & \leq \sum_{i=1}^n \mu(\hat{s}_i \mid s') u(s' \cap \hat{s}_i).  
    \end{split}
    \end{equation*}
    \item [ii)] $\hat{s} \cap s'$ induces action $a_k$, for some $k \neq \ell$.  
    \begin{equation*}
    \begin{split}
        V(\hat{\pi} \mid s') & = \sum_{s'' \in \hat{\pi}} \mu(s'' \mid s') \, u(s' \cap s'') - u(s') \\
        & = \sum_{i=1}^n \mu(\hat{s}_i \mid s') u(s' \cap \hat{s}_i) + \mu(\hat{s} \mid s') \, u(s' \cap \hat{s})  - u(s') \\
        & = \sum_{i=1}^n \mu(\hat{s}_i \mid s') u(s' \cap \hat{s}_i) + \mu(\hat{s} \mid s') \,\mu(\omega_k \mid \hat{s} \cap  s') \, u_k  -  \mu(\omega_\ell \mid s') \, u_\ell \\
        & \leq \sum_{i=1}^n \mu(\hat{s}_i \mid s') u(s' \cap \hat{s}_i) + \mu(\omega_k \mid s') \, u_k  -  \mu(\omega_\ell \mid s') \, u_\ell \\
        & \leq \sum_{i=1}^n \mu(\hat{s}_i \mid s') u(s' \cap \hat{s}_i). 
    \end{split}
    \end{equation*}
\end{itemize}
From both cases we have that 
\begin{equation*}
V(\hat{\pi} \mid s) \leq \sum_{i=1}^n \mu(\hat{s}_i \mid s') u(s' \cap \hat{s}_i).
\end{equation*}
Since $\mu(\omega_i \mid s' \cap \hat{s}_i) = 1$ and consequently $\mu(\hat{s}_i \mid s', \omega_j) = 0$ whenever $j \neq i$, it holds that
    \begin{equation*}
    \begin{split}
     V(\hat{\pi} \mid s')  \leq \sum_{i=1}^n \mu(\hat{s}_i \mid s') u(s' \cap \hat{s}_i) =  \sum_{i=1}^n \mu(\hat{s}_i \mid s') \, u_i  = \sum_{i=1}^n \mu(\omega_i \mid s') \, \mu(\hat{s}_i \mid \omega_i, s') \, u_i. 
    \end{split}
    \end{equation*}
Finally, from (\ref{eqgeneral}) it follows that 
\begin{equation*}
\begin{split}
V(\hat{\pi} \mid \pi') & = \sum_{s' \in \pi'} \mu(s') \, V(\hat{\pi} \mid s') \\
& \leq \sum_{s' \in \pi'} \mu(s')\, \sum_{i=1}^n \mu(\omega_i \mid s') \, \mu(\hat{s}_i \mid \omega_i, s') \, u_i \\
& = \sum_{i=1}^n \sum_{s' \in \pi'} \mu(s') \, \mu(\omega_i, \hat{s}_i \mid s') u_i \\
 & = \sum_{i=1}^n  \mu(\omega_i, \hat{s}_i) u_i  \\
  & = \sum_{i=1}^n  \mu(\omega_i) \, \mu(\hat{s}_i \mid \omega_i) u_i \\
  & = V(\hat{\pi} \mid \pi). 
\end{split}
\end{equation*}
\end{proof}

\subsection{Proof of Proposition ~\ref{generalproposition}} \label{appendix:generalproposition}

\begin{proof}
Consider the following menu 
\begin{equation*}
    \mathcal{M} = \{(\hat{\pi}, t(\pi)
    ) \mid \pi \in \Pi\} \quad \text{ where } t(\pi) = \overline{U} -U(\pi). 
\end{equation*}
Note that this is the maximum possible revenue the data broker can achieve, since each type ends up with full information and pays their maximum valuation for that information. Lemma~\ref{generallemma} shows that this menu is incentive compatible, as no type benefits from purchasing another type's signal-price pair. 
\end{proof}

\newpage

\end{document}